\documentclass[11pt]{amsart}
\usepackage[english]{babel}
\usepackage[utf8]{inputenc}
\usepackage[T1]{fontenc}
\usepackage{amsmath, amssymb, amsfonts}
\usepackage{amsthm}
\usepackage{mathrsfs}
\usepackage{enumitem}
\usepackage{mathtools}
\usepackage{todonotes}
\usepackage{standalone}
\usepackage{pgfplots}
\pgfplotsset{compat=1.17}
\usepackage{thmtools, thm-restate}
\usepackage{booktabs}
\usepackage{hyperref}
\hypersetup{hidelinks}
\usepackage[capitalise, noabbrev]{cleveref}
\usepackage[font=small]{caption}
\usepackage{microtype}
\DeclarePairedDelimiter\ceil{\lceil}{\rceil}

\DeclarePairedDelimiter\abs{\lvert}{\rvert}

\usepackage{bm}           
\usepackage{soul}
\usepackage{graphicx}
\graphicspath{{./figures/}}
\usepackage{tikz}
\usepackage{tikz-cd}

\usepackage[vlined, ruled,linesnumbered]{algorithm2e}
\usepackage{algpseudocode}

\newcommand{\R}{\mathbb{R}} 

\let\epsilon\varepsilon

\newcommand{\Ball}{\text{Ball}}
\DeclareMathOperator{\Nrv}{Nrv}
\DeclareMathOperator{\V}{\mathcal{V}}
\DeclareMathOperator{\Vor}{Vor}
\DeclareMathOperator{\DelT}{Del}
\newcommand{\DC}{\mathcal{DC}}
\newcommand{\NN}{\mathcal{N}}
\DeclareMathOperator{\Del}{\mathcal{D}}
\DeclareMathOperator{\Tel}{Tel}
\DeclareMathOperator{\TelVor}{\mathcal{T}}

\DeclareMathOperator{\TelOffset}{\mathcal{TO}}

\DeclareMathOperator{\TelVO}{\mathcal U}
\DeclareMathOperator{\TelDC}{\mathcal W}

\DeclareMathOperator{\DelCech}{\mathcal{DC}}
\DeclareMathOperator{\Grid}{\mathrm{Grid}}

\DeclareMathOperator{\im}{im}
\DeclareMathOperator{\colim}{colim}

\DeclareMathOperator*{\minimize}{minimize}

\newcommand{\Top}{\textbf{Top}}

\PassOptionsToPackage{x11names}{xcolor}
\usepackage{xcolor}
\definecolor{darkred}{RGB}{139,0,0}
\definecolor{darkblue}{RGB}{0,0,139}
\definecolor{darkgreen}{RGB}{0,100,0}

\newcommand{\deff}{\emph}

\newcommand{\Offset}[0]{\mathcal{O}}

\newcommand{\Incr}[0]{\mathcal{I}}

\newcommand{\St}{\mathrm{Star}}

\providecommand\given{}
\newcommand\SetSymbol[1][]{%
  \nonscript\:#1\vert
  \allowbreak
  \nonscript\:
  \mathopen{}}
\DeclarePairedDelimiterX\Set[1]\{\}{%
  \renewcommand\given{\SetSymbol[]}
  #1
}

\makeatletter
\DeclareRobustCommand{\subto}{%
  \mathrel{\mathpalette\short@to\relax}%
}

\newcommand{\short@to}[2]{%
  \mkern2mu
  \clipbox{{.35\width} 0 0 0}{$\m@th#1\vphantom{+}{\rightarrow}$}%
  }
\makeatother

\makeatletter
\DeclareRobustCommand{\subleft}{%
  \mathrel{\mathpalette\short@left\relax}%
}

\newcommand{\short@left}[2]{%
  \mkern2mu
  \clipbox{0 0 {.35\width} 0}{$\m@th#1\vphantom{+}{\leftarrow}$}%
  }
\makeatother

\makeatletter
\DeclareRobustCommand{\subup}{%
  \mathrel{\mathpalette\short@up\relax}%
}

\newcommand{\short@up}[2]{%
  \mkern2mu
  \clipbox{0 {.35\height} 0 0}{$\m@th#1\vphantom{+}{\uparrow}$}%
}
\makeatother

\makeatletter
\DeclareRobustCommand{\subdown}{%
  \mathrel{\mathpalette\short@down\relax}%
}
\newcommand{\short@down}[2]{%
  \mkern2mu
  \clipbox{0 0 0 {.45\height}}{$\m@th#1\vphantom{+}{\downarrow}$}%
}
\makeatother

\makeatletter
\DeclareRobustCommand{\subupright}{%
  \mathrel{\mathpalette\short@near\relax}%
}
\newcommand{\short@near}[2]{%
  \mkern2mu
  \clipbox{{.55\width} {.55\height} 0 0}{$\m@th#1\vphantom{+}{\nearrow}$}%
}
\makeatother

\makeatletter
\DeclareRobustCommand{\subdownleft}{%
  \mathrel{\mathpalette\short@sw\relax}%
}
\newcommand{\short@sw}[2]{%
  \mkern2mu
  \clipbox{0 0 {.55\width} {.55\height}}{$\m@th#1\vphantom{+}{\swarrow}$}%
}
\makeatother

\renewcommand{\subparagraph}[1]{\subsubsection*{\textbf{\textup{#1}}}}

\theoremstyle{plain}
\newtheorem{theorem}{Theorem}[section]
\newtheorem{lemma}[theorem]{Lemma}
\newtheorem{corollary}[theorem]{Corollary}
\newtheorem{proposition}[theorem]{Proposition}

\theoremstyle{definition}
\newtheorem{definition}[theorem]{Definition}
\newtheorem{example}[theorem]{Example}

\theoremstyle{remark}
\newtheorem{remark}[theorem]{Remark}

\newcommand{\SizeBound}{{\lceil (d+1)/2 \rceil + 1}}
\newcommand{\CompBound}{{\lceil d/2\rceil +2}}

\title{Bifunction and Interlevel Delaunay Trifiltrations}

\author[A.J.~Alonso]{\'Angel Javier Alonso}
\address{CUNEF Universidad, Madrid, Spain and Institute of Geometry, Graz University of Technology, Austria}
\email{angel.alonso@cunef.edu}

\author[M.~Kerber]{Michael Kerber}
\address{Institute of Geometry, Graz University of Technology, Austria}
\email{kerber@tugraz.at}

\author[T.~Lam]{Tung Lam}
\address{University at Albany, SUNY, USA}
\email{tlam@albany.edu}

\author[M.~Lesnick]{Michael Lesnick}
\address{University at Albany, SUNY, USA}
\email{mlesnick@albany.edu}

\author[A.~Rathod]{Abhishek Rathod}
\address{Ben Gurion University, Israel}
\email{arathod@post.bgu.ac.il}

\thanks{A.~Alonso: Austrian Science Fund (FWF) grants 10.55776/P33765 and 10.55776/W1230.
M.~Kerber: Austrian Science Fund (FWF) grant 10.55776/P33765.
T.~Lam and M.~Lesnick: Simons Foundation Award 963845.
A.~Rathod: European Research Council (ERC) grant PARAPATH (101039913).}

\thanks{Software available at \url{https://bitbucket.org/mkerber/function_delaunay}.
Dataset available at \url{https://doi.org/10.5281/zenodo.19227801}.}

\begin{document}

\maketitle

\begin{abstract}
A key property of the Delaunay filtration is that it is topologically (i.e., weakly) equivalent to the offset
  (union-of-balls) filtration. Recently, this filtration has been extended to
  point clouds equipped with an $\mathbb{R}$-valued function, yielding a
  computable 2-parameter filtration that satisfies an analogous weak equivalence.                                                                              Motivated in part by the study of time-varying data, we introduce a
  3-parameter extension of the Delaunay filtration for point clouds equipped with
  an $\mathbb{R}^2$-valued function, also satisfying an analogous weak
  equivalence. For a point cloud $X \subset \mathbb{R}^d$, our trifiltration has
  size $O\bigl(|X|^{\lceil(d+1)/2\rceil+1}\bigr)$. We present an algorithm that
  computes this trifiltration in time
  $O\bigl(|X|^{\lceil d/2\rceil+2}\bigr)$,
  together with an implementation. Our experiments demonstrate that the
  implementation can handle thousands of points in $\mathbb{R}^3$, with memory growth that is nearly linear.
\end{abstract} 

\section{Introduction}
\subparagraph{Background.}
Given a finite point cloud $X \subset \R^d$, the \emph{offset filtration} of $X$ is the nested sequence of spaces $ \Offset(X) = \left(\Offset(X)_r \right)_{r \in \R^+}$, where $\R^+=[0,\infty)$ and $\Offset(X)_r$ is the union of closed balls of radius $r$ centered at the points of $X$ (\Cref{fig:del-offset}). In topological data analysis (TDA), one often computes the persistent homology of $\Offset(X)$.  Directly computing $\Offset(X)$ is difficult, so one usually instead computes the \emph{Delaunay filtration} $\mathcal D(X)$ (also called the $\alpha$-filtration) \cite{Edelsbrunner:alphashape,edelsbrunner2022computational} (\Cref{fig:del-offset}), a filtration of the Delaunay triangulation $\DelT(X)$ of $X$ that is topologically equivalent to $\Offset(X)$ (i.e., weakly equivalent, see \cref{def:top_equiv}), and therefore has the same persistent homology.  For low-dimensional Euclidean data, $\mathcal D(X)$ is the standard choice of filtration in TDA, because its computation is far more efficient than alternatives like the Rips and \v Cech filtrations. Bauer and Edelsbrunner \cite{bauer2017morse} showed that $\Offset(X)$ is also weakly equivalent to a variant of $\mathcal D(X)$ called the \emph{Delaunay-\v Cech} filtration and denoted $\mathcal{DC}(X)$, where the birth index of each simplex in the filtration is the radius of its smallest enclosing ball.

\begin{figure}[h!]
\begin{center}
\centering
\includegraphics[trim={0.3cm 0.5cm 0.3cm 0.5cm},clip,page=1,scale=0.63]{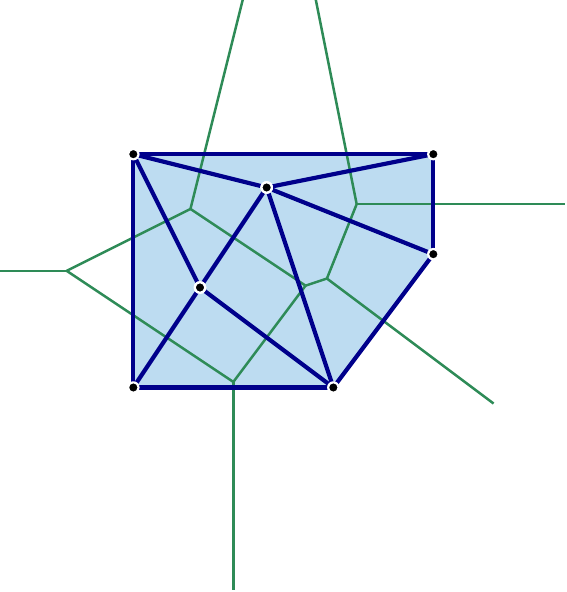}
\hfill
\includegraphics[trim={0.3cm 0.5cm 0.3cm 0.5cm},clip,page=2,scale=0.63]{del-vor}
\end{center}
\caption{Left: For a point set $X \subset \R^2$, the Voronoi decomposition $\Vor(X)$ (green) and the corresponding Delaunay triangulation $\DelT(X)$ (dark blue). Right: The offset $\Offset(X)_r$ (orange) and Delaunay complex $\Del(X)_r \subset \DelT(X)$ for a fixed radius $r$.}
\label{fig:del-offset}
\end{figure}

In many TDA applications, the point cloud $X$ comes equipped with a function $\delta\colon X \to \R$ \cite{chazal2011scalar,carlsson2009theory}.  It is then natural to seek a refinement of the persistent homology of $X$ that is sensitive to the structure of this function.  Functions considered in the TDA literature include (co)density, eccentricity (centrality), discrete curvature \cite{carlsson2009theory}, partial charge \cite{cang2020persistent,cang2018representability}, and time \cite{lin2022temporalTDA}. 
To extend the persistence pipeline to point clouds equipped with multiple functions, Carlsson and Zomorodian \cite{carlsson2009theory} proposed \emph{multiparameter persistence}, and specifically, the \emph{sublevel offset bifiltration} $\Offset(\delta) = \left( \Offset(\delta)_{q,r} \right)_{(q,r)\in \R\times \R^{+}}$, where \[\Offset(\delta)_{q,r}=\Offset(\{x\in X\mid \delta(x)\leq q\})_{r}.\]
The sublevel {\v C}ech and sublevel Rips bifiltrations are defined analogously, and can be computed in essentially the same way as the usual \v Cech and Rips filtrations.  

Defining an analogous sublevel Delaunay bifiltration is more subtle because of the non-monotonicity of Delaunay complexes with respect to insertion of points.  
However, recent work~\cite{alonsoDelaunayBifiltrationsFunctions2024} introduced sublevel Delaunay and sublevel Delaunay-\v Cech bifiltrations that are both weakly equivalent to the sublevel offset bifiltration and amenable to efficient computation.  Experiments with a publicly available implementation showed that, in practice, computing the sublevel Delaunay-\v Cech bifiltration is only modestly more expensive than computing the ordinary Delaunay filtration.  The work~\cite{alonsoDelaunayBifiltrationsFunctions2024} is part of a large body of recent work aiming to realize multiparameter persistence as a computationally viable data analysis tool; for an introduction, see \cite{botnan2023introduction}.  

\subparagraph{Contributions.}
In this paper, we extend the results of \cite{alonsoDelaunayBifiltrationsFunctions2024} on sublevel-Delaunay bifiltrations to functions $\gamma\colon X \to \R^2$, where $X\subset \R^d$ is in general position. Specifically, we introduce the \emph{sublevel Delaunay} and \emph{sublevel Delaunay-\v{C}ech trifiltrations} of $\gamma$, denoted $\Del(\gamma)$ and $\DelCech(\gamma)$, respectively. We show that both trifiltrations have size $O(|X|^{\SizeBound})$, can be computed in time $O(|X|^\CompBound)$, and are weakly equivalent to the \emph{sublevel offset trifiltration}, defined as follows:
 \begin{definition}
The \deff{sublevel offset trifiltration} of $\gamma$, denoted $\Offset(\gamma)$, is the collection of Euclidean subspaces $\bigl(\Offset_{p,r}\bigr)_{(p,r)\in \mathbb{R}^2 \times \mathbb{R}^{+}}$,
where
$\Offset_{p,r} \coloneqq \Offset(X_p)_r$ and 
$X_p \coloneqq \{x \in X \mid \gamma(x) \le p\}.$

Here $\le$ denotes the product partial order on $\mathbb{R}^2$.\end{definition}
Our approach to defining and computing the trifiltrations $\Del(\gamma)$ and $\DelCech(\gamma)$ parallels the approach of \cite{alonsoDelaunayBifiltrationsFunctions2024} for $\R$-valued functions, but is necessarily more involved.  To illustrate the problem, consider four point clouds $X_a\subset X_b,X_c\subset X_d$, as in \cref{fig:non-monotone}.  While for any $r\in \R^+$, the offsets satisfy ${\Offset(X_a)_r\subset \Offset(X_b)_r,\Offset(X_c)_r\subset \Offset(X_d)_r}$, these containment relations do not hold for the corresponding Delaunay triangulations, nor for the Delaunay complexes of radius $r$.  The construction of \cite{alonsoDelaunayBifiltrationsFunctions2024} applies only when the point clouds are totally ordered by containment, hence does not apply here, because neither $X_b$ nor $X_c$ contains the other.  

\begin{figure}
  \centering
  \includegraphics{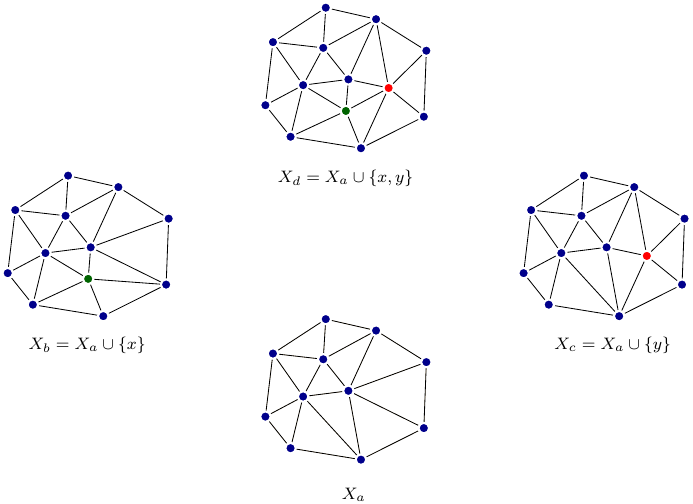}
  \caption{Non-monotonicity of Delaunay triangulations of planar point clouds $X_a$, $X_b=X_a\cup \{x\}$, $X_c=X_a\cup \{y\}$, and $X_d=X_a\cup \{x,y\}$.  The points $x$ and $y$ are shown in green and red.}
    \label{fig:non-monotone}
\end{figure}

    In analogy to the previous work on Delaunay bifiltrations, the crux of our approach is the definition and computation of the maximum complex (i.e., colimit) of the trifiltrations $\Del(\gamma)$  and $\DelCech(\gamma)$. This complex is the same for both trifiltrations; we call it the \emph{incremental complex} and denote it as $\Incr$ (\cref{def:biincremental}).  The complex $\Incr$ has vertex set $X$, is $(d+2)$-dimensional, and contains the union of the Delaunay triangulations of all sublevel sets of $\gamma$ as a subcomplex.  
    
Our main size and complexity bounds hinge on a characterization of $\Incr$ in terms of \emph{conflict pairs} (see \cref{def:conflicts,lem:quasi-purity}), which extends a corresponding result in the bifiltered setting  \cite[Lemma 2.3]{alonsoDelaunayBifiltrationsFunctions2024}.  Assuming that $d$ is constant, we use our characterization of $\Incr$ to show that $\Incr$ has size $O(\abs{X}^{\SizeBound})$.  We also derive a simple algorithm to compute $\Incr$ in time $O(\abs{X}^{\ceil{d/2}+2})$, based on the standard Bowyer-Watson algorithm for incremental computation of Delaunay triangulations.  We then present optimizations to this algorithm, exploiting the fact that Delaunay triangulations are \emph{local}, in a sense made precise by \cref{lem:local-star-star}.  

For $\sigma\subset X$ non-empty, let $\gamma(\sigma)=\bigvee_{x\in \sigma} \gamma(x)$, where $\bigvee$ denotes the \emph{join}, i.e., coordinate-wise maximum. 
Given $\Incr$, the trifiltration $\Del(\gamma)$ is defined by assigning each simplex $\sigma \in \Incr$ the birth index $(\gamma(\sigma), \omega_{\sigma})$, where $\omega_\sigma$ is the \emph{incremental Delaunay radius}, defined in \cref{sec:bifuncdel}.
$\DelCech(\gamma)$ is defined by assigning each $\sigma \in \Incr$ the birth index $(\gamma(\sigma), m_\sigma)$, where $m_\sigma$ is the radius of the smallest enclosing ball of $\sigma$.  We show that the birth indices of all simplices in $\Del(\gamma)$ can be computed in time $O(|\Incr|)$, by a variant of the standard algorithm for computing the birth indices in a Delaunay filtration.\footnote{The analogue of this result for Delaunay bifiltrations was briefly mentioned in \cite{alonsoDelaunayBifiltrationsFunctions2024} but not explained, and has not yet appeared in print.  While we explicitly treat only the trifiltered case, the algorithm and proof are nearly identical in the bifiltered case.}  Similarly, we can compute the birth indices of all simplices in $\DelCech(\gamma)$ in time $O(|\Incr|)$ by computing the smallest enclosing ball of each $\sigma\in\Incr$ in constant time with a standard algorithm~\cite{fischerFastSmallestEnclosingBallComputation2003,fischerSmallestEnclosingBall2004,cgal:fghhs-bv-24a}.  

To show that $\Del(\gamma)$, $\DelCech(\gamma)$, and $\Offset(\gamma)$ are topologically equivalent, we use an argument based on 2-parameter mapping telescopes and a functorial version of the nerve theorem \cite{bauerUnifiedViewFunctorial2023}.  The corresponding result in the bifiltered setting \cite[Corollary 3.4]{alonsoDelaunayBifiltrationsFunctions2024} was proven by extending a discrete Morse theory argument of Bauer and Edelsbrunner \cite{bauer2017morse} for 1-parameter persistence.  While we expect that approach to extend to our trifiltered setting, we find the nerve-based approach used in this paper to be more intuitive. 

We implement our algorithm to compute $\DelCech(\gamma)$ and test this implementation on point clouds in $\R^2$ and $\R^{3}$ with up to 16000 points; our implementation is publicly available~\cite{function_delaunay}.    In our experiments, the memory usage scales nearly linearly, while runtime scales nearly quadratically, suggesting that substantially larger computations are feasible.  While our current implementation does not exploit parallelization, our algorithm is highly parallelizable.

\subparagraph{Motivation.}
While there is a substantial recent literature on computing bifiltrations of point cloud and metric data in TDA \cite{corbet2023computing,alonsoDelaunayBifiltrationsFunctions2024,blaser2024core,lesnick2024nerve, alonso2024sparse,alonso2023filtration}, our work apparently is the first to seriously engage with the algorithmic aspects of trifiltrations.  We hope that our work will lower the barrier to practical data analysis with 3-parameter persistence, which has seen very limited use so far.  

Our specific motivation for introducing computable Delaunay-type models of sublevel-offset trifiltrations is threefold:  First, we wish to enable persistence analysis which simultaneously considers a pair of functions on a point cloud, e.g., density and eccentricity.  
Indeed, if one is interested in doing a topological analysis of a function on a point cloud with noise or non-uniformities in density, then it is important to take density into account, and one natural approach is to treat density as an additional filtration parameter.

Second, the definition of the sublevel-offset bifiltration involves a choice of filtering direction (i.e., whether we filter by sublevel or superlevel sets). While for some functions $\delta\colon X\to \R$, like codensity, it is natural to fix a direction, for other functions like time or partial charge, it can be unnatural to fix a filtering direction, and one would like a more symmetric construction. To this end, one can consider the \emph{interlevel-offset trifiltration} of $\delta$, defined as follows:
\begin{definition}
For a function $\delta: X \to \R$, the \deff{interlevel-offset trifiltration}
of $\delta$, denoted $\Offset^{\updownarrow}(\delta)$, is the collection of Euclidean subspaces
$
(\Offset^{\updownarrow}(\delta)_{p, r})_{(p, r) \in \R^2 \times \R^+}$, where $\Offset^{\updownarrow}(\delta)_{p, r} \coloneqq \Offset\bigl(\{x \in X \mid  -p_1 \leq \delta(x) \leq p_2\}\bigr)_r$.

\end{definition}
Note that $\Offset^{\updownarrow}(\delta)$ is a special case of the sublevel-offset trifiltration, namely $\Offset^{\updownarrow}(\delta)=\Offset(\gamma)$, where $\gamma=(-\delta,\delta)$. We are especially interested in computing the homology of the interlevel-offset trifiltration in the case that $\delta$ is a time function, as this could potentially be useful in the analysis of dynamic data such as that arising in the study of collective motion of animals \cite{adams2020topological}, in dynamical systems, or in the sliding-window analysis of time series \cite{perea2015sliding}.
While $\Offset^{\updownarrow}(\delta)$ is only one of several possible choices of a multifiltration that can be built from time-varying data (e.g., see \cite{kim2021spatiotemporal,strommen2023topological}), it is notable for its flexibility---unlike some alternatives, the definition does not require a consistent labeling of the points at each time step, and it accommodates different numbers of points at different time steps.

Our third motivation is the following: The \emph{degree-Rips bifiltration} has emerged in recent years as a popular choice of density-sensitive bifiltration on metric data \cite{rolle2024stable,blumberg2024stability,lesnick2015interactive}. The construction is appealing in part because it requires no parameter choice, whereas the sublevel-Rips filtration of a (co)density function requires a choice of bandwidth parameter. In the same way, one can define parameter-free \emph{degree-offset} and \emph{degree-\v{C}ech bifiltrations}. We are interested in the problem of computing an analogous \emph{degree-Delaunay bifiltration} which is topologically equivalent to the degree-offset bifiltration and more computable than degree-Rips or degree-\v{C}ech for low-dimensional data. In ongoing work on this problem, we noticed that the problems of computing the sublevel Delaunay(-\v{C}ech) trifiltration and the degree-Delaunay bifiltration are very closely related. As computing the former is simpler yet already quite involved, we view the results of this paper as a key step towards the development of computationally efficient degree-Delaunay bifiltrations. However, degree-Delaunay bifiltrations will not be studied in this paper.

\subparagraph{Other Related Work.}
Our work and the prior work \cite{alonsoDelaunayBifiltrationsFunctions2024} solve particular instances of a general computational problem: Given a filtered point cloud $Z=(Z_p)_{p\in P}$ indexed by a poset $P$, one has an associated offset filtration indexed by the product poset $P\times [0,\infty)$; the problem is to define and compute a generalized Delaunay filtration of reasonable size which is weakly equivalent to this offset filtration.  The case that $Z$ is of the form $X \hookrightarrow X\cup Y \hookleftarrow Y$ was studied by Reani and Bobrowski~\cite{reaniCoupledAlphaComplex2023}, who call the resulting generalized Delaunay filtration the \emph{coupled $\alpha$-filtration}.  
An extension of this construction to $n\geq 2$ 
point clouds and their union, called \emph{chromatic $\alpha$-complexes}, was introduced in \cite{dimontesanoChromaticAlphaComplexes2024,biswasSizeChromaticDelaunay2022}. 

Many works have studied time-varying point cloud or metric data through the lens of persistence.  To avoid the difficulties of multiparameter persistence, such work most often either considers persistence with respect to a scale parameter and ignores persistence across the time parameter \cite{munch2013:thesis,adams2020topological, adams:crockerplot,guzel2022detecting,morozov2008:homological, cohen2006vines}, or fixes a scale parameter and considers extended or zigzag persistence across time \cite{tymochko2020:zigzagbifurcation, kim2020analysis, kim2024extracting}.  In particular, the problem of extending Delaunay complexes to time-varying data has been studied in \cite{kerber20133d,edelsbrunner2012medusa}  using 1-parameter (extended) persistent homology with a fixed radius parameter. 
It is natural to work in the multiparameter setting and consider persistence with respect to both scale and time simultaneously.
This was previously explored by Kim and Mémoli \cite{kim2021spatiotemporal}, who introduce and study a Rips trifiltration for time-varying metric data with consistent labels across time, and by Saunders \cite{Saunders2025ZigzagPH} and Dey and Samaga \cite{dey2026quasizigzagpersistencetopological}, who study a version of multiparameter persistence 
indexed over a product of a zigzag poset and a totally ordered set.

Finally, we note that a construction similar to the one used in \cite{alonsoDelaunayBifiltrationsFunctions2024} was previously used by Sheehy to construct sparse Delaunay filtrations \cite{sheehy:LIPIcs.SoCG.2021.58}.  

\subparagraph{Outline.} 
\cref{Sec:Preliminaries} covers preliminaries used throughout the paper.  \Cref{sec:bifuncdel}  defines the incremental Delaunay complex $\Incr$ of a function $\gamma\colon X\to\R^{2}$, thereby completing the definition of the trifiltrations  $\Del(\gamma)$ and $\DelCech(\gamma)$.    \Cref{Sec:Conflict_Pairs} introduces conflict pairs and presents our characterization of $\Incr$ in terms of these.  \Cref{sec:computation} applies this characterization to obtain an algorithm for computing $\Incr$, and also discusses optimizations to this algorithm.  \Cref{sec:top_equiv} establishes that $\Del(\gamma)$ and $\DC(\gamma)$ are weakly equivalent to the sublevel offset trifiltration $\Offset(\gamma)$. 
\Cref{sec:implement}  describes our implementation and reports the results of our experiments.  Appendix~\ref{sec:radiusDel} gives our algorithm to compute incremental Delaunay radii.  
\section{Preliminaries}\label{Sec:Preliminaries}
A \emph{sphere} is a set of the form $\{w\in \R^d\mid \|w-x\|=r\}$ for some choices of $x\in \R^d$ and $r\geq 0$.  In particular, for $x\in \R^d$, we regard $\{x\}$ as a sphere of radius 0.
Given a set $Q\subset \R^d$, we say a sphere $S\subset \R^d$ is a \emph{circumsphere} of $Q$ if all points of $Q$ lie on $S$, and we say  $S$ is \emph{$Q$-empty} if no point of $Q$ lies inside $S$.  

We fix a finite set $X\subset\R^{d}$ in \emph{general position}, i.e., every non-empty subset $Q \subset X$ with at most $d+1$ points is affinely
independent, and no point of $X \setminus Q$ lies on the smallest circumsphere
of $Q$.  The \emph{Delaunay triangulation} of $X$ is the simplicial complex \[\DelT(X)=\{\sigma\subset X\mid \sigma \textup{ is non-empty and has an $X$-empty circumsphere}\}.\]  
  The \deff{Voronoi cell} of a point $x\in X$ is the region
  \[\Vor(x, X) = \{w \in \R^d \mid  \|w - x\| \leq \|w - y\| \textup{ for all } y\in X\}.\]
It is easily checked that $\DelT(X)$ is equal to the \emph{nerve} of the collection of Voronoi cells $\left\{\Vor(x, X) \mid x\in X \right\}$; see \cref{sec:top_equiv} for the definition of the nerve.  

For $x\in X$ and $r\geq 0$, the \deff{Voronoi ball} $\V(x,X)_r$ is the region
 \[\V(x,X)_r = \Vor(x, X) \cap \Ball_r(x),\]
 where $\Ball_r(x)\coloneqq \{w\in \R^d \mid \|x-w\|\leq r\}$ is the closed ball of radius $r$ centered at $x$.     
Letting $\Del(X)_r$ denote the nerve of the collection of Voronoi balls $\left\{\V(x,X)_r \right\}_{x\in X}$, we call $\Del(X)\coloneqq (\Del(X)_r)_{r\in \R^+}$ the \emph{Delaunay filtration} of $X$.

We also fix a function $\gamma = (\gamma_1, \gamma_2) \colon X\to \R^{2}$. 
For simplicity, throughout the paper, we assume that both $\gamma_1$ and $\gamma_2$ are injective.  
This assumption entails no loss of generality: one can perturb $\gamma$ to enforce such injectivity, carry out our trifiltration constructions for the perturbed function, and then readily recover correct constructions for the unperturbed function. 
For $i=1,2$, let $<_i$ denote the total order on $X$ induced by $\gamma_{i}$, and  
let $\max_{i}(\sigma)$ be  the maximum element of $\sigma$ with respect to $<_i$.

Let $\Grid(\gamma)=\im \gamma_1\times \im \gamma_2$, and consider $p=(p_1, p_2)\in \Grid(\gamma)$.  For $p_1<\max(\im \gamma_1)$, let $p{\subto}=(p_1^+,p_2)$, where $p_1^+$ is the element of $\im \gamma_1$ immediately after $p_1$.  
We define the variants of this notation $p{\subup}$, $p{\subleft}$, and $p{\subdown}$ in the analogous way.  For $p\in \Grid(\gamma)$ with $p_1<\max(\im \gamma_1)$ and $p_2<\max(\im \gamma_2)$, let $p{\subupright}=(p{\subup}){\subto}$, and define the notation $p{\subdownleft}$ analogously.

Throughout the paper, we consider the product partial order on $\R^2$, given by $p\leq q$ if and only if $p_1\leq q_1$ and $p_2\leq q_2$.  We say $p,q$ are \emph{comparable} if $p \leq q$ or $q \leq p$; otherwise, they are \emph{incomparable}.  For $p\in\R^{2}$, let $X_{p}\coloneqq \Set{x\in X\given \gamma(x) \leq p}$ denote the $p$-sublevel set of $\gamma$.

\section{Delaunay trifiltrations} \label{sec:bifuncdel}
\begin{definition}\label{def:biincremental}
  The \deff{incremental Delaunay complex} $\Incr$ 
  is the simplicial complex whose simplices are the non-empty subsets $\sigma
  \subset X$ for which there exists a sphere $S$ such that 
    \begin{itemize}
    \item $S$ is a circumsphere of $\sigma \setminus
  \{x,y\}$, where $x=\max_1(\sigma)$ and $y=\max_2(\sigma)$,
        \item $x$ and $y$ are either inside or on $S$,
        \item $S$ is $(X_{\gamma(\sigma)}\setminus \{x,y\})$-empty.  
    \end{itemize}
We call $S$ a \emph{witness} of $\sigma$.
\end{definition}
Note that in \cref{def:biincremental}, we can have $\max_1(\sigma)=\max_2(\sigma)$, in which case $\gamma(\sigma)=\gamma(\max_1(\sigma))$.

\begin{proposition}\label{Prop:Min_Rad_Witness}
Any simplex $\sigma\in \Incr$ has a unique smallest (i.e., minimum-radius) witness.
\end{proposition}

\begin{proof}
As above, we write $x=\max_1(\sigma)$ and $y=\max_2(\sigma)$.  We first assume that $\tau \coloneqq \sigma \setminus \{x, y\}$ is non-empty.  Choose an arbitrary point $z\in \tau$.  There exists a unique smallest witness of $\sigma$ if and only if the following optimization problem has a unique solution:
\begin{equation}\label{eq:opt_Prob}
\begin{aligned}
\minimize_{c \in A} \quad & \|c - z\|^2 \\
\text{subject to} \quad & \|c - x\| \leq \|c - z\|, \\
& \|c - y\| \leq \|c - z\|, \\
& \|c - w\| \geq \|c - z\| \quad \text{for all } w \in X_{\gamma(\sigma)} \setminus \{x, y\},
\end{aligned}
\end{equation}
where \[A \coloneqq \{c\in \R^d : \|c-w\| = \|c-z\| \; \textup{ for all } w\in \tau\}.\]  Note that $A$ is an affine subspace 
of dimension $d-\dim(\tau)$.  The objective function of this problem is \emph{coercive} (i.e., $\lim_{\|c\|\to \infty}  \|c - z\|^2=\infty$) and strictly convex.  In addition, the feasible set ${F}$ of this problem is a closed convex set 
in $A$, and ${F}$ is non-empty since $\sigma \in \Incr$.  Thus, by a standard result of convex optimization theory \cite[Section 3.1]{bertsekas2009convex}, the minimizer of this optimization problem exists and is unique.

Now assume that  $\tau = \emptyset$.  If $x = y$, then the unique smallest witness of $\sigma=\{x\}$ is $\{x\}$.  If $x\ne y$, then for any witness $S$ of $\sigma$ with $x$ and $y$ not both lying on $S$, there is a witness $S'$ of $\sigma$ with smaller radius such that $x$ and $y$ lie on $S'$.  (To obtain $S'$, we first shrink $S$ while keeping its center fixed, until a point of $\sigma$, say $x$, touches the sphere.  We then further shrink the sphere, keeping the point $x$ on the sphere and moving the sphere's center towards $x$, until $y$ also touches the sphere.)  This implies that $\sigma\in \Del(X_{\gamma(\sigma)})$ and that  $\sigma$ has a unique smallest witness if and only if $\sigma$ has a unique smallest $X_{\gamma(\sigma)}$-empty circumsphere.  It is well known that any simplex of a Delaunay triangulation has a unique smallest empty circumsphere; this can be verified using a convex optimization argument similar to the one above.  As  $\sigma\in \Del(X_{\gamma(\sigma)})$, it follows that $\sigma$ has a unique smallest $X_{\gamma(\sigma)}$-empty circumsphere, so $\sigma$ has a unique smallest witness.
\end{proof}

\begin{definition}
For $\sigma\in \Incr$, the \emph{incremental Delaunay radius of $\sigma$},  denoted $\omega_\sigma$, is the minimum radius of a sphere witnessing $\sigma$. 
 \end{definition}

The last part of the proof of \cref{Prop:Min_Rad_Witness} establishes the following useful result.

\begin{proposition}\label{Prop:Empty_Tau_Delaunay}
For $\sigma\in \Incr$ with $\dim(\sigma)=1$ and $\max_1(\sigma)\ne \max_2(\sigma)$, the smallest witness of $\sigma$ is a circumsphere of $\sigma$.  In particular, $\sigma\in \Del(X_{\gamma(\sigma)})$.  
\end{proposition}

\begin{definition}
For $\gamma: X\to \R^2$, the \deff{sublevel Delaunay trifiltration} $\Del(\gamma)$ and \deff{sublevel Delaunay-\v{C}ech trifiltration} $\DelCech(\gamma)$ are defined over $\R^{2}\times\R^+$ as follows:
\begin{align*}
    \Del(\gamma)_{p, r} & \coloneqq \Set{\sigma\in\Incr \mid  \gamma(\sigma)\leq p\text{ and }  \omega_\sigma \leq r}, \\
    \DelCech(\gamma)_{p, r} & \coloneqq \Set{\sigma\in\Incr  \mid  \gamma(\sigma)\leq p\text{ and }  m_\sigma \leq r },
\end{align*}
where $m_\sigma$ is the radius of the smallest enclosing ball of $\sigma$.
\end{definition}

\begin{proposition}\label{prop:contains3}
For $r \geq 0$ and $p\in \R^2$, we have $\Del(X_p)_r\subset \Del(\gamma)_{p, r}$.
\end{proposition}
\begin{proof}
If $\sigma\in \Del(X_p)_r$, then there exists
an $X_p$-empty circumsphere $S$ of $\sigma$ with radius at most $r$.
The same sphere $S$ witnesses that $\sigma\in \Incr$, and hence that $\sigma\in \Del(\gamma)_{p,r}$. 
\end{proof}

Taking $r$ sufficiently large, \Cref{prop:contains3} implies that $\Incr$ contains the Delaunay triangulation $\DelT(X_p)$ for all $p\in \R^2$.

\section{Conflict pairs and triples}\label{Sec:Conflict_Pairs}
We now characterize the simplices of $\Incr$ in a way that is conducive to computation.

\begin{definition}[Conflict pairs and triples]\label{def:conflicts} \mbox{}
\begin{itemize}
\item[(i)]
For $p\in\Grid(\gamma)$,  a \emph{conflict pair at $p$} is a pair $(\sigma,x)$ such that 
\begin{enumerate}
\item  $\sigma$ is a $d$-simplex in $\DelT(X_p)$,
\item either 
\begin{itemize}
\item[{}] $(p_1<\max(\im \gamma_1)$ and $x\in X_{p\subto})$ or
\item[{}] $(p_2<\max(\im \gamma_2)$ and $x\in X_{p\subup})$,
\end{itemize}
\item $x$ lies inside the circumsphere of $\sigma$.
\end{enumerate}

\item[(ii)] For distinct conflict pairs $(\sigma, x)$ and $(\sigma, y)$ at $p$, we call $(\sigma, x,y)$ a \deff{conflict triple at $p$}.
\item[(iii)] We call $(\sigma, x)$ a \deff{conflict pair} if it is a conflict pair at some $p$, and similarly for conflict triples.
\end{itemize}
\end{definition}

\cref{def:conflicts} is illustrated in \cref{fig:conflict_triple}.  Note that if $(\sigma, x,y)$ is a conflict triple at $p$, then by the injectivity of $\gamma_1$ and $\gamma_2$, either $x\in X_{p{\subto}}$ and $y\in X_{p{\subup}}$, or else $y\in X_{p{\subto}}$ and $x\in X_{p{\subup}}$.

\begin{figure}[htpb]
    \centering
    \includegraphics[width=0.50\textwidth, page=1]{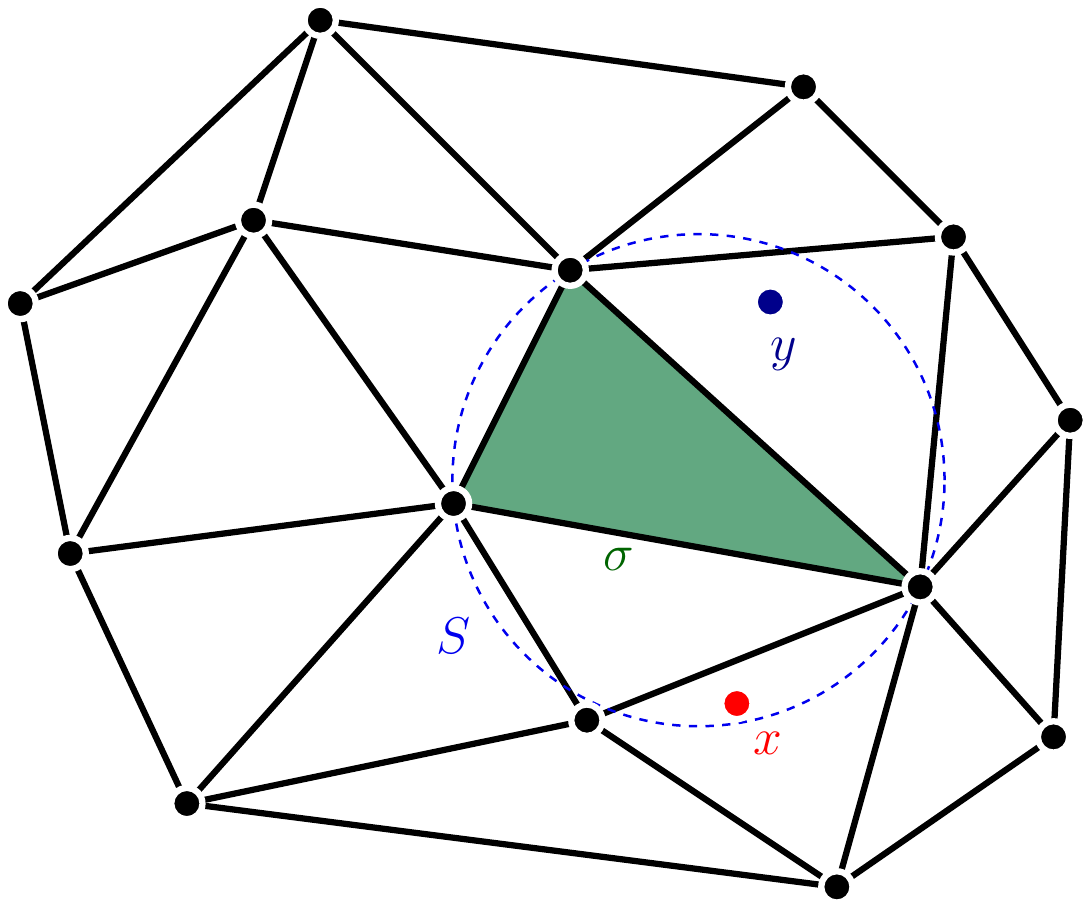}
    \hfill
    \includegraphics[width=0.40\textwidth, page=2]{conflict}%
    \caption{Illustration of a conflict triple.  Let $\sigma \in \DelT(X_{(2,1)})$ be a 2-simplex with circumsphere $S$, as shown in green on the left.  Assume that $\gamma(\sigma)=(2,1)$  and that $x$ (red) and $y$ (blue) are points of $X$ lying inside $S$, with $\gamma(x)=(3,2)$ and $\gamma(y)=(1,3)$.  Then letting $p=(2,2)$, the injectivity of $\gamma_2$ implies that $X_{(2,1)}=X_{p}$, so $\sigma \in \Del(X_p)$.  In addition, $x\in X_{p\subto} = X_{(3,2)}$, so $(\sigma, x)$ is a conflict pair at $p$.  Similarly, $y\in X_{p\subup} = X_{(2,3)}$, so $(\sigma, y)$ is also a conflict pair at $p$.  Since $(\sigma, x)$ and $(\sigma, y)$ are distinct conflict pairs at $p$, $(\sigma, x, y)$ is a conflict triple at $p$.}
    \label{fig:conflict_triple}
\end{figure}

\begin{lemma}\label{lem:incomparable_conflicts}
    If $(\sigma,x)$ and $(\sigma,y)$ are distinct conflict pairs, then $\gamma(x)$ and $\gamma(y)$ are incomparable.
\end{lemma}

\begin{proof}
To arrive at a contradiction, suppose that $\gamma(x)<\gamma(y)$ and that $(\sigma,y)$ is a conflict pair at $p$.  Then either $y\in X_{p{\subto}}$ or $y\in X_{p{\subup}}$, so by the injectivity of $\gamma_1$ and $\gamma_2$, we have $\gamma(x)\leq p.$  Since $(\sigma,x)$ is a conflict pair, we have $\sigma\not\in \DelT(\gamma(x))$, hence $\sigma\not\in\DelT(X_p)$.  This contradicts the assumption that $(\sigma,y)$ is a conflict pair at $p$. 
\end{proof}

\Cref{lem:incomparable_conflicts} implies that for any $d$-simplex $\sigma$ in $\bigcup_p \DelT(X_p)$, the set  \[T_{\sigma}\coloneqq\{x \mid (\sigma,x)\textup{ is a conflict pair}\}\] is totally ordered by $\gamma_1$-coordinate.  Henceforth, we assume that $T_{\sigma}$ is given this order.

\begin{lemma}\label{lem:consec_order_conflict_triple}
For $x,y\in T_{\sigma}$, $(\sigma,x,y)$ is a conflict triple if and only if $x,y$ are ordered consecutively in $T_{\sigma}$.
\end{lemma}

\begin{proof}
If $x,y$ are ordered consecutively, then it is easily checked that $(\sigma,x,y)$ is a conflict triple at $(\gamma(x) \vee \gamma(y)){\subdownleft}$.  Conversely, if there exists $z\in T_{\sigma}$ with $x<z<y$ in the order on $T_{\sigma}$, then for any $p$ with $x\in X_{p\subup}$ and  $y\in X_{p\subto}$, we have $z\in X_{p}$, so $\sigma\not \in \DelT(X_p)$.  Hence $(\sigma,x,y)$ is  not a conflict triple.   
\end{proof}

The next lemma follows immediately from the definition of $\Incr$.

\begin{lemma}\label{lem:conflicts}\mbox{}
\begin{enumerate}
\item[(i)] If $(\sigma, x)$ is a conflict pair, then  $\sigma \cup \{x\}$  is a $(d+1)$-simplex in $\Incr$. 
\item[(ii)] If $(\sigma, x, y)$ is a conflict triple, then $\sigma \cup \{x, y\}$ is a $(d+2)$-simplex in~$\Incr$.
\end{enumerate}
\end{lemma}

We refer to the $(d+1)$-simplices and  $(d+2)$-simplices of \cref{lem:conflicts} as \emph{conflict simplices}.  In what follows, we show that each simplex of $\Incr$ is a face of a conflict simplex, provided $X$ has the following property.  

\begin{definition}
We say that $X$ has the \deff{$\Delta$-property} if the total orders $<_1$
and $<_2$ coincide on the first $d+1$ points and $X$ is contained in
the convex hull of these $d+1$ points. 
\end{definition}

If $X$ has the $\Delta$-property, then we let $\Delta$ denote the set of the first $d+1$ points.  

The $\Delta$-property can be easily enforced in practice by simply adding $d+1$ artificial points ``at
infinity.''  The analogous condition for a single total order (or a
slight variant thereof) is commonly exploited in the incremental computation of
Delaunay triangulations \cite{edelsbrunnerIncrementalTopologicalFlipping1996,
  boissonnat2009incremental,boissonnat2000triangulations}, and is also used in
the previous work on sublevel-Delaunay
bifiltrations~\cite{alonsoDelaunayBifiltrationsFunctions2024}. 

\begin{proposition}
\label{lem:quasi-purity} If $X$ has the $\Delta$-property and
  $|X| \geq d+2$, then every $\tau \in \Incr( \gamma)$ is either
  \begin{itemize}[noitemsep,topsep=0pt]
    \item a face of a $(d+2)$-dimensional conflict simplex or
    \item a face of a $(d+1)$-dimensional conflict simplex $\sigma \cup \{x\}$ where $|T_{\sigma}|=1$.
  \end{itemize}
\end{proposition}

\cref{lem:quasi-purity} is an analogue of \cite[Lemma 2.3]{alonsoDelaunayBifiltrationsFunctions2024} and has a similar but more involved proof. Our proof  will use the following two technical lemmas.

\begin{lemma}\label{Lem:Conflict_Technicality_1}
Given $Z\subset \R^d$ in general position, $\eta\in \DelT(Z)$, and $z\in Z\setminus \eta$, if there exists a $(Z\setminus \{z\})$-empty circumsphere $S^+$ of $\eta$ such that $z$ lies inside or on $S^+$, then $\eta\cup\{z\}\in \DelT(Z)$.  

\end{lemma}

\begin{proof}
Since $\eta\in \DelT(Z)$, there exists a $Z$-empty circumsphere $S^-$ of $\eta$.  Note that $S^-$ is a $(Z\setminus \{z\})$-empty circumsphere of $\eta$ such that $z$ lies on or outside $S^-$.  Let $W\subset \R^d$ denote the set of all centers of $(Z\setminus \{z\})$-empty circumspheres of $\eta$.  We have $W=\bigcap_{y\in \eta} \Vor(y, Z\setminus \{z\})$, so $W$ is convex, hence path connected.  Consider the function $\operatorname{rad}\colon W\to \R$ that maps each element of $W$ to the radius of the associated circumsphere of $\eta$, and note that $\operatorname{rad}$ is continuous.  Therefore, by applying the intermediate value theorem to the function $g\colon W\to \R$ given by $g(w)=\|z-w\|-\operatorname{rad}(w)$, we see that there exists a $(Z\setminus \{z\})$-empty circumsphere $S$ of $\eta$ such that $z$ lies on $S$.  Then $S$ is a $Z$-empty circumsphere of $\eta\cup\{z\}$, so $\eta\cup\{z\}\in \DelT(Z)$.  
\end{proof}

\begin{lemma}\label{Lem:Conflict_Technicality_2}
If $X$ has the $\Delta$-property, $\tau\in \DelT(X_{\gamma(\tau)})$, and $\tau\not \subset \Delta$, then $\tau$ is a face of a conflict simplex.
\end{lemma}

\begin{proof}
We use the following notation:
\begin{itemize}
\item $x=\max_1(\tau)$, 
\item $q=\gamma(\tau)$, 
\item $\mu$ is a $d$-dimensional coface of $\tau$ in $\DelT(X_{q})$,
\item $\nu=\mu\setminus \{x\}$.  
\end{itemize}
Note that $\nu\in \DelT(X_{q{\subleft}})$.   By the $\Delta$-property, $\nu$ has the same number of $d$-dimensional cofaces in $\DelT(X_{q{\subleft}})$ and in $\DelT(X_{q})$, namely, either one (when $\nu\subset \Delta$) or two such cofaces.  Thus, since $\mu\in \DelT(X_{q})$ but $\mu\notin \DelT(X_{q{\subleft}})$, there must exist a $d$-dimensional coface $\sigma$ of $\nu$ in $\DelT(X_{q{\subleft}})$ such that $\sigma \notin \DelT(X_{q})$.  Then $(\sigma, x)$ is a conflict pair at $q{\subleft}$ and $\tau$ is a face of the conflict simplex~$\sigma\cup\{x\}$.   
\end{proof}

\begin{proof}[Proof of \cref{lem:quasi-purity}]
By \cref{lem:consec_order_conflict_triple}, if $\sigma \cup \{x\}$ is a $(d+1)$-dimensional conflict simplex with $|T_{\sigma}|>1$, then $(\sigma, x)$ is a face of some $(d+2)$-dimensional conflict simplex $\sigma \cup \{x, y\}$.  Thus, it suffices to show that every $\tau \in \Incr$ is a face of a conflict simplex.

\medskip\noindent\textbf{Case 1: $\tau\subset \Delta$.}\\
Let $x$ be a minimal element of $X\setminus \Delta$ with respect to the product partial order on $\R^2$.  By the $\Delta$-property, $x$ lies in the convex hull of $\Delta$, so the circumsphere
of $\Delta$ contains $x$ in its interior. Since $x \notin \Delta$, we
have $\gamma_1(x) > \gamma_1(\Delta)$ and $\gamma_2(x) > \gamma_2(\Delta)$.  Thus, 
$(\Delta, x)$ is a conflict pair at both $\gamma(x){\subdown}$ and $\gamma(x){\subleft}$, and $\tau$ is a face of the conflict simplex $\Delta\cup\{x\}$.

\medskip\noindent\textbf{Case 2:  $\max_1(\tau)=\max_2(\tau)\notin \Delta$.}\\
Our argument here is essentially the same as the one used to prove \cite[Lemma 2.3]{alonsoDelaunayBifiltrationsFunctions2024}.  We give a streamlined version which avoids proof by contradiction.  

Let $x=\max_1(\tau)$, $\eta=\tau\setminus\{x\}$, and $q=\gamma(\tau)$.  If $\dim(\tau)=0$ (i.e., $\tau=\{x\}$), then the $\Delta$-property and general position assumption together imply that  $\{x\}$ lies in the interior of some $d$-simplex $\sigma\in \DelT(X_{q{\subleft}})$, so $(\sigma,x)$ is a conflict pair at $q{\subleft}$ and $\tau$ is a face of the conflict simplex $\sigma\cup\{x\}$.  If $\dim(\tau)>0$, then $\eta\ne\emptyset$ and $\eta\in \DelT(X_{q{\subleft}})$.  If $\eta\notin \DelT(X_{q})$, let $\sigma\in \DelT(X_{q{\subleft}})$ be a $d$-dimensional coface of $\eta$.  Then $\sigma\notin \DelT(X_{q})$, so $(\sigma, x)$ is a conflict pair at $q{\subleft}$ and $\tau$ is a face of the conflict simplex $\sigma\cup\{x\}$.  If $\eta\in \DelT(X_{q})$, then since $\tau\in \Incr$, \cref{Lem:Conflict_Technicality_1} implies that $\tau\in \DelT(X_{q})$, so \cref{Lem:Conflict_Technicality_2} implies that $\tau$ is a face of a conflict simplex.  

\medskip\noindent\textbf{Case 3:  $\max_1(\tau)\ne\max_2(\tau)$}.\\
Let $x=\max_1(\tau)$ and $y=\max_2(\tau)$.  Note that neither $x$ nor $y$ is in $\Delta$.  Let $\eta=\tau\setminus\{x,y\}$ and $q=\gamma(\tau)$.  If $\dim(\tau)=1$ (i.e., $\tau=\{x,y\}$), then \cref{Prop:Empty_Tau_Delaunay} tells us that $\tau\in \DelT(X_q)$, so \cref{Lem:Conflict_Technicality_2} implies that $\tau$ is a face of a conflict simplex.  If $\dim(\tau)>1$, then $\eta\ne\emptyset$ and $\eta\in \DelT(X_{q{\subdownleft}})$.  First, suppose that $\eta\notin \DelT(X_{q{\subleft}})$ and $\eta\notin \DelT(X_{q{\subdown}})$.  Let $\sigma\in \DelT(X_{q{\subdownleft}})$ be a $d$-dimensional coface of $\eta$.    Note that $\sigma\not\in \DelT(X_{q{\subleft}})$ and $\sigma\notin \DelT(X_{q{\subdown}})$.  Thus, $(\sigma, x)$ and $(\sigma, y)$ are both conflict pairs at $q{\subdownleft}$, and $\tau$ is a face of the conflict simplex $\sigma\cup\{x,y\}$.  

We next show that if $\eta\in \DelT(X_{q{\subdown}})$, then $\tau$ is a face of a conflict simplex; the symmetric argument shows that if $\eta\in \DelT(X_{q{\subleft}})$, then $\tau$ is a face of a conflict simplex.  Suppose that $\eta\in \DelT(X_{q{\subdown}})$.  Since $\tau\in \Incr$, there exists an $(X_{q}\setminus \{x,y\})$-empty circumsphere $S^+$ of $\eta$ with $x$ and $y$ lying inside or on $S^+$.  Since $X_{q{\subdown}}= X_{q}\setminus \{y\}$, we see that $S^+$ is in particular an $(X_{q{\subdown}}\setminus \{x\})$-empty circumsphere of $\eta$ with $x$ lying inside or on $S^+$.  Therefore, \cref{Lem:Conflict_Technicality_1} implies that $\eta\cup\{x\}\in \DelT(X_{q{\subdown}})=\DelT(X_{q}\setminus \{y\})$.

We claim that there exists either 
\begin{enumerate}
\item an $(X_q\setminus \{y\})$-empty circumsphere $S$ of $\eta\cup\{x\}$ with $y$ inside or on $S$, or else 
\item an $(X_q\setminus \{x\})$-empty circumsphere $S$ of $\eta\cup\{y\}$ with $x$ inside or on $S$.
\end{enumerate}

To prove the claim, note that since $\eta\cup\{x\} \in \DelT(X_{q}\setminus \{y\})$, there exists an $(X_q\setminus \{y\})$-empty circumsphere $S^-$ of $\eta\cup\{x\}$.   If $y$ lies inside or on $S^-$, then the first condition of the claim holds with $S=S^-$, so assume that $y$ lies outside $S^-$.  Note that each of $S^+$ and $S^-$ is an $(X_q\setminus \{x,y\}$)-empty circumsphere of $\eta$ with $x$ lying inside or on the sphere.  Letting $W$ denote the set of centers of all such spheres, we have \[W=\left(\bigcap_{z\in \eta} \Vor(z, X_q\setminus \{x,y\})\right) \cap \Vor(x,X_q\setminus \{y\}),\] so $W$ is convex, hence path connected.    Since $y$ lies inside or on $S^+$ and $y$ lies outside $S^-$, an intermediate value theorem argument similar to the one of \cref{Lem:Conflict_Technicality_1} shows that there exists an $(X_q\setminus \{x,y\})$-empty circumsphere $S$ of $\eta$ with $y$ lying on $S$ and $x$ inside or on $S$; equivalently, $S$ satisfies the second condition of the claim.  This proves the claim.  
   
We next show that if the first condition of the claim holds, then $\tau$ is a face of a conflict simplex.  If $\eta\cup\{x\}\not \in \DelT(X_{q})$, then let  $\sigma\in \DelT(X_{q{\subdown}})$ be a $d$-dimensional coface of $\eta\cup\{x\}$.  Then $\sigma\not\in \DelT(X_{q})$, so $(\sigma,y)$ is a conflict pair at $q{\subdown}$ and  $\tau$ is a face of the conflict simplex $\sigma\cup \{y\}$.  If $\eta\cup\{x\}\in \DelT(X_{q})$, then assuming that the first condition of the claim holds, \cref{Lem:Conflict_Technicality_1} implies that $\tau\in \DelT(X_q)$, so \cref{Lem:Conflict_Technicality_2} implies that $\tau$ is a face of a conflict simplex.

If the second condition of the claim holds, then we have in particular that $\eta\cup\{y\}\in  \DelT(X_{q}\setminus \{x\})=\DelT(X_{q{\subleft}})$, so essentially the same argument (replacing $q{\subdown}$ with $q{\subleft}$ and exchanging $x$ with $y$) shows that $\tau$ is a face of a conflict simplex.
\end{proof}

\subparagraph{Size.}
We apply the above results, together with a result about the incremental construction of Delaunay triangulations~\cite{edelsbrunnerIncrementalTopologicalFlipping1996}, to bound the size of~$\Incr$:

\begin{theorem}\label{thm:incr-size}
 For constant dimension $d$, $\Incr$ has $O(\abs{X}^{\SizeBound})$ simplices.
\end{theorem}

\begin{proof}
By~\cref{lem:quasi-purity} and the discussion above, the number of simplices in $\Incr$ is $O(C)$, where $C$
is the number of conflict pairs and conflict triples.  Moreover, by \cref{lem:consec_order_conflict_triple} the number of conflict triples is at most the number of conflict pairs.  Therefore, to prove the result, it suffices to bound the number of conflict pairs.  
By a result in \cite[Section 9]{edelsbrunnerIncrementalTopologicalFlipping1996} (see also~\cite[Theorem 2.6]{alonsoDelaunayBifiltrationsFunctions2024}), the total number of conflict pairs along any fixed vertical or horizontal line is $O(|X|^{\ceil{(d+1)/2}})$.  Thus, the total number of conflict pairs is $O(\abs{X}^{\ceil{(d+1)/2} + 1})$.
\end{proof}

\section{Computation}\label{sec:computation}
By \cref{lem:quasi-purity}, to compute  $\Incr$ it suffices to compute all conflict pairs and conflict triples.  
In the following, we explain how to compute all conflict pairs.   Once we have done so, we can easily compute all conflict triples via \cref{lem:consec_order_conflict_triple}, by iterating through the set of conflict pairs $T_{\sigma}$ for each $d$-simplex $\sigma\in \bigcup_{p\in \Grid(\gamma)} \DelT(X_p)$.    
  
\subparagraph{The Bowyer-Watson algorithm.}  
We make essential use of the \deff{Bowyer-Watson
algorithm}~\cite{bowyerComputingDirichletTessellations1981,watsonComputingNdimensionalDelaunay1981}, a standard incremental algorithm for computing Delaunay triangulations.  Given the Delaunay triangulation  $\DelT(Z)$ of a point cloud $Z\subset \R^d$  and a point $z\in {\R^d}$, the Bowyer-Watson algorithm 
computes $\DelT(Z \cup \{z\} )$.  To do so, it first identifies each $d$-simplex $\sigma$ such that $z$ lies inside the circumsphere of $\sigma$.  Each such simplex $\sigma$ is then removed from $\DelT(Z)$, resulting in a star-convex
region $\Sigma$ centered at $z$. Finally, to obtain $\DelT(Z \cup \Set{z})$, the region $\Sigma$ is retriangulated by forming the simplicial cone with base the boundary of $\Sigma$ and apex $z$ 
(See~\cref{fig:BWalgorithm}).  For $\sigma$ a $d$-simplex as above, we call the pair $(\sigma,z)$ a \deff{BW-conflict (for $Z$)}.  

\begin{figure}[htpb]
    \centering
    \includegraphics[width=0.28\textwidth,page=1]{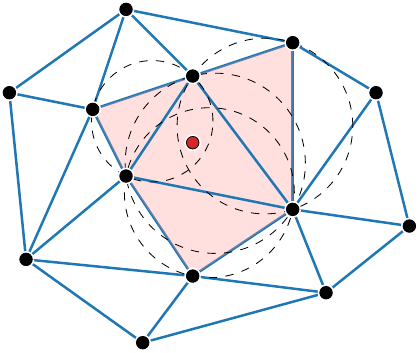}
    \hfill
    \includegraphics[width=0.28\textwidth, page=2]{bowyer_watson.pdf}
    \hfill
    \includegraphics[width=0.28\textwidth, page=3]{bowyer_watson.pdf}
    \caption{
        The Delaunay triangulation $\DelT(Z)$ (blue), and the new point $z$
        (red). Triangles whose circumcircles contain $z$ are shaded.  We remove
        all these triangles, leaving an untriangulated star-convex region $\Sigma$ centered at $z$.
$\Sigma$ is then retriangulated by the simplicial cone with base the boundary of $\Sigma$ and apex $z$, yielding $\DelT(Z \cup \Set{z})$.
}
    \label{fig:BWalgorithm}
\end{figure}

\begin{remark}
Note that the definition of BW-conflict is distinct from, but closely related to, that of a conflict pair given in \cref{def:conflicts}\,(i).  Namely, every conflict pair $(\sigma,x)$ at $p$ is a BW-conflict for $X_p$, and conversely, if $x\in X_{p{\subto}}\cup X_{p{\subup}}$ and $(\sigma,x)$ is a BW-conflict for $X_p$, then $(\sigma,x)$ is a conflict pair.  
\end{remark}

\subparagraph{Naive Computation of Conflict Pairs.}
To compute all conflict pairs, a natural strategy is to traverse each vertical line in $\Grid(\gamma)$ in increasing order, iteratively applying the Bowyer-Watson algorithm, and to then do the same for each horizontal line.  To explain this, we focus on the traversal of vertical lines, as the case of horizontal lines is symmetric.  A naive version of the strategy, which already satisfies our main complexity bound (\cref{Thm:Time_Bound}), is as follows: For each $a\in \im \gamma_1$ in increasing order, iterate through $\im \gamma_2=\{b_1,\ldots,b_n\}$, also in increasing order.  For each $i < n$, let $p=(a,b_i)$.   Given $\DelT(X_{p})$,  if $X_{p{\subup}}\setminus X_{p}$ is non-empty, hence equal to $\{x\}$ for some $x\in X$, then we apply Bowyer-Watson to compute both $\DelT(X_{p{\subup}})$ and all conflict pairs of the form $(\sigma,x)$ at $p$.  

\begin{theorem}\label{Thm:Time_Bound}
The above algorithm computes $\Incr$ in time $O(\abs{X}^{\ceil{d/2}+2})$.
\end{theorem}
\begin{proof}
The algorithm inserts a point into a Delaunay triangulation $O(\abs{X}^{2})$ times.  As observed in~\cite[Section 2]{alonsoDelaunayBifiltrationsFunctions2024},  the cost of each such insertion is bounded by the number of $d$-simplices in the triangulation, which in this case is $O(\abs{X}^{\ceil{d/2}})$.  Thus the cost of computing all conflict pairs is $O(\abs{X}^\CompBound)$.  

If for each $d$-simplex $\sigma\in \bigcup_{p\in \Grid(\gamma)} \DelT(X_p)$, we are given the set $T_{\sigma}$ in order, then we can compute all conflict triples in time \[O\left(\sum_{\sigma} |T_{\sigma}|\right)=O(\abs{X}^{\SizeBound})=O(\abs{X}^\CompBound).\]  Thus, to finish the proof, it suffices to observe that the algorithm discovers the elements of each set $T_{\sigma}$ in order, so that no additional sorting is required.  To see this, note that as we process the vertical lines, we discover all elements of $T_{\sigma}$ in order,  except possibly the last element.  Indeed, the last element $y$ of $T_{\sigma}$ is discovered as we traverse the vertical line containing $\gamma(y)$ if and only if $\gamma_2(y)> \max_{x\in \sigma} \gamma_2(x)$.  If $y$ is not discovered as we traverse this vertical line, then it is discovered afterwards, when we traverse the horizontal line containing $\gamma(y)$.  Either way, $y$ is the last element of $T_{\sigma}$ that we discover.
\end{proof}

\subparagraph{Computing trifiltrations.} 
Once $\Incr$ is computed, completing the computation of either $\Del(\gamma)$ or
$\DelCech(\gamma)$ amounts to computing the birth index of each simplex $\sigma\in\Incr$ in the trifiltration.
Recall that the birth index of $\sigma$ in $\DelCech(\gamma)$ is $(\gamma(\sigma), m_\sigma)$, where $m_\sigma$ is the radius of the smallest enclosing ball of $\sigma$.  Computation of $\gamma(\sigma)$  is trivial.  In general, smallest enclosing balls can be computed via standard algorithms with efficient implementations that, assuming  $d$ constant, run in constant time on point sets in $\R^d$ of constant size~\cite{welzlSmallestEnclosingDisks1991,fischerFastSmallestEnclosingBallComputation2003,fischerSmallestEnclosingBall2004,gartnerFastRobustSmallest1999,cgal:fghhs-bv-24a} .  Thus, $m_\sigma$ can be computed in constant time because $\dim(\sigma)\leq d+2$.  Further, recall that the birth index of $\sigma$ in $\Del(\gamma)$ is $(\gamma(\sigma), \omega_\sigma)$, where $\omega_\sigma$ is the incremental Delaunay radius of \Cref{def:biincremental}.  In Appendix \ref{sec:radiusDel}, we give an algorithm to compute $\omega_\sigma$  for all $\sigma\in\Incr$ in time $O(|\Incr|)$.  From this and \cref{Thm:Time_Bound}, we obtain the following:

\begin{corollary}\label{Thm:Filtration_Time_Bound}
Both  $\Del(\gamma)$ and $\DelCech(\gamma)$ can be computed in time $O(\abs{X}^{\ceil{d/2}+2})$.
\end{corollary}

\subsection{Optimized computation of conflict pairs}
\label{Sec:Optimizations}
The above approach to computing conflict pairs involves considerable redundant computation across different lines.  We next discuss strategies for reducing this redundancy.    Whereas the unoptimized algorithm described above uses only point insertions in Delaunay triangulations, our optimizations use both insertions and deletions.  Our aim in selecting these optimizations is to reduce  the total number $\zeta(X)$ of insertions and deletions, which can be seen heuristically as a proxy for the total cost of the algorithm in practice.  Our optimizations do not reduce $\zeta(X)$ in the worst case, and do not improve the asymptotic worst-case complexity.
But in some classes of examples they yield an asymptotic reduction of $\zeta(X)$.

In what follows, we consider only optimizations of the computations along vertical lines, as the case of horizontal lines is symmetric.  First, we note that for any $x\in X$, it suffices to begin the traversal of the vertical line containing $\gamma(x)$ at index $p=\gamma(x){\subdown}$, rather than at index $(\gamma_1(x),\min_{y\in X}\gamma_2(y))$.  However, to exploit this idea, one must select an algorithm to compute $\DelT(X_{\gamma(x){\subdown}})$.  A simple approach is to use the Bowyer-Watson algorithm with randomized insertion of points, which is known to be highly efficient in expectation \cite{Clarkson1989}.  However, this requires the same number of point insertions as the unoptimized algorithm, which constructs $\DelT(X_{\gamma(x){\subdown}})$ via point insertions in order of $\gamma_2$-value.  
We next introduce an alternative algorithm to compute $\DelT(X_{\gamma(x){\subdown}})$ which exploits the work done along the previous vertical line to reduce $\zeta(X)$ in favorable cases.

\subparagraph{Computing $\DelT(X_{\gamma(x){\subdown}})$ via deletions.}
Consider $x\in X$, and let $W$ denote the final sublevel set of $\gamma$ encountered just before we traverse the vertical line containing $\gamma(x)$, i.e., $W=\{w\in X\mid \gamma_1(w)<\gamma_1(x)\}$.  Noting that $X_{\gamma(x){\subdown}}\subset W$, let $R= W \setminus X_{\gamma(x){\subdown}}.$
We compute $\DelT(X_{\gamma(x){\subdown}})$ from $\DelT(W)$ by iteratively removing the points of $R$, using the standard algorithm of \cite{Devillers1999} to repair the Delaunay triangulation after each removal.  

The number of vertex removals required to obtain $\DelT(X_{\gamma(x){\subdown}})$ from $\DelT(W)$ is at most the number of insertions required for the subsequent traversal of the vertical line containing $\gamma(x)$, because each vertex removed must be reinserted.  Thus, in all instances, this strategy increases $\zeta(X)$ by at most a factor of two, compared to the unoptimized algorithm.  Further, the next example shows that this strategy can lead to an asymptotic reduction in $\zeta(X)$ on a certain family of inputs.   

\begin{example}
If the product partial order on $\R^2$ restricts to a total order on $\gamma(X)$, then $\zeta=O(|X|)$ for the optimized strategy described above, but $\zeta=O(|X|^2)$ for the unoptimized strategy.  On the other hand, if the elements of $X$ are pairwise incomparable, then $\zeta=O(|X|^2)$ for both strategies.  
\end{example}

\subparagraph{Reducing the number of insertions and deletions.} 
In fact, to compute the conflict pairs first appearing on the vertical line containing $\gamma(x)$, it is not necessary to explicitly remove all points of $R$, but only a subset $Q\subset R$.  This idea leads to an additional optimization, which we call the \deff{local algorithm}, that can further reduce $\zeta(X)$.  To distinguish the approaches, we call the optimized algorithm described above the \deff{non-local algorithm}.

The local algorithm is illustrated in \cref{fig:geometric_optimization}.  The algorithm proceeds as follows:  We maintain a Delaunay triangulation $T$, to be modified by point removals and insertions.  First, we compute $T=\DelT(W\cup \{x\})$ from $\DelT(W)$ via the Bowyer-Watson algorithm.  We next iteratively remove from $T$ the neighbors of $x$ that belong to $R$, in arbitrary order, until $x$ has no neighbor in $R$.  Letting $Q$ denote the set of removed points, after these removals, we have $T=\DelT((W\cup\{x\})\setminus Q)$.  We then remove $x$ from $T$.  Finally, we insert the points of $Q\cup \{x\}$ into $T$ in increasing order of $\gamma_2$-value, noting the BW-conflicts that arise.  Each BW-conflict $(\sigma, z)$ such that $x\in \sigma\cup\{z\}$ is recorded as a conflict pair.

\begin{figure}[h]
    \centering
    \includegraphics[scale=0.6, page=1]{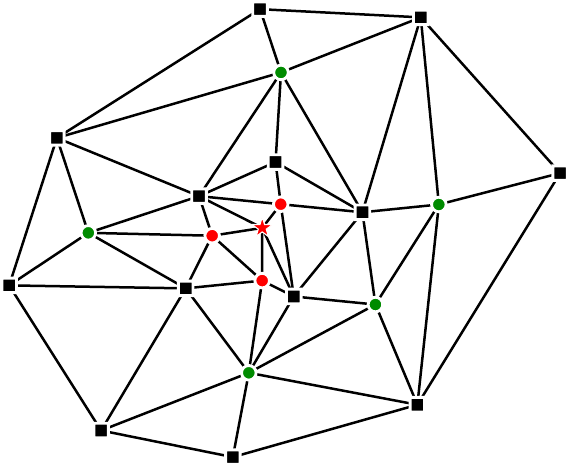}%
    \hfill
    \includegraphics[scale=0.6, page=2]{alg_visualization}
  \caption{
After inserting $x$ into $\DelT(W)$ to obtain $\DelT(W\cup\{x\})$, the first phase of the local algorithm iteratively removes points of $R=W\setminus X_{\gamma(x){\subdown}}$ that are Delaunay neighbors of $x$, 
until there are no such neighbors.  The left figure shows $\DelT(W\cup \{x\})$ and  the right figure shows $\DelT(W\cup \{x\} \setminus Q)$, where $Q$ is the set of removed points.  The point $x$ is shown as a red star, while the points of $Q$, $R\setminus Q$, and $X_{\gamma(x){\subdown}}$ are shown as red circles, green circles, and squares, respectively. 
}
\label{fig:geometric_optimization}
\end{figure}

The correctness of this algorithm hinges on the fact that Delaunay triangulations are local, in the sense of the following lemma.  For a simplicial complex $K$ and vertex $v\in K$, the \deff{neighborhood} of $v$ in $K$, denoted $\NN(v,K)$, is the set of vertices incident to $v$ in the 1-skeleton of $K$, together with $v$ itself.  The \deff{star} of $v$ in $K$, denoted $\St(v, K)$, is the set of simplices of $K$ that contain $v$.

\begin{lemma}\label{lem:local-star-star}
  For  $Z\subset\R^d$ a finite point set in general position and $z\in Z$, let $\NN_z=\NN(z, \DelT(Z))$.
For any $Y$ such that $\NN_z \subset Y \subset Z$, we have
  \begin{equation*}
    \St(z, \DelT(\NN_z)) = \St(z, \DelT(Y)).
  \end{equation*}
\end{lemma}

\begin{proof}
We claim that $\Vor(z,\NN_z)=\Vor(z,Y)=\Vor(z,Z)$. Indeed, \[\Vor(z,Z)=\bigcap_{y\in \NN_z} \{w\in \R^d\mid \|w-z\|\leq \|w-y\|\}=\Vor(z,\NN_z).\]
Since $\NN_z \subset Y \subset Z$, we have $\Vor(z,Z)\subset \Vor(z,Y)\subset \Vor(z,\NN_z)$, and the claim follows.

If $\sigma\in\St(z, \DelT(Y))$, then there is a $Y$-empty circumsphere $S$ of $\sigma$.  Since $\NN_z\subset Y$, in fact $S$ is also $\NN_z$-empty. To establish that $\sigma\in\St(z, \DelT(\NN_z))$, it remains to check that $\sigma\subset\NN_z$. Because $S$ is a $Y$-empty circumsphere of $\sigma$ and $z\in \sigma$, its center $c$ lies in $\Vor(z, Y) = \Vor(z, Z)$.  Thus, $S$ is also $Z$-empty, so~$\sigma\subset\NN_z$.

Conversely, if $\sigma\in\St(z, \DelT(\NN_z))$, then there is a circumsphere $S$ of $\sigma$ that is $\NN_z$-empty.  Letting $c$ denote the center of $S$, we have $c\in \Vor(z, \NN_z)=\Vor(z, Y)$.  Thus, $S$ is also $Y$-empty, so $\sigma\in\St(z,\DelT(Y))$.
\end{proof}

\begin{proposition}\label{lem:local-reinsert} 
The local algorithm correctly computes all conflict pairs.
\end{proposition}

\begin{proof} 
    It suffices to show that for any $x\in X$, the set of conflict pairs first recorded while traversing the vertical line containing $\gamma(x)$ is the same for both algorithms.  Note that for either algorithm and any such conflict pair $(\sigma, z)$, we have $x\in \sigma\cup \{z\}$; for the local algorithm, this holds by definition, while for the non-local algorithm, it holds because if $x\not\in \sigma\cup \{z\}$, then the conflict pair $(\sigma, z)$ would have been recorded along a previous vertical line.   
   
   Let us write $R= \{z_1 <_2 z_2 <_2 \cdots <_2 z_m\}.$
    For $k\in \{0,\dots,m\}$, let $N_k$ and $L_k$ be the point sets maintained by the 
    non-local and local algorithms, respectively, after the first $k$ points of $R$ have been inserted, i.e., 
    \[
    N_k = X_{ \gamma(x)} \cup \{z_1,\dots,z_{k}\},
    \qquad \textup{ and } \qquad
    L_k = N_k \cup (R\setminus Q).
    \]
    We claim that for each $k\in \{0,\ldots, m\}$, 
    \begin{equation}\label{eq:star-invariant} 
        \St(x,\DelT(N_k))=\St(x,\DelT(L_k)).  
    \end{equation}
    To prove the claim, we check that $\NN(x, \DelT(L_k)) \subset N_k$; the claim then follows from~\cref{lem:local-star-star} by taking $z = x$, $Y=N_k$ and $Z = L_k$.  Let  $w \in \NN(x, \DelT(L_k))$.   If $w \notin L_0=X_{ \gamma(x)}\cup (R\setminus Q)$, then $w \in L_k \setminus L_0 = \{z_1, \dots, z_{k}\} \subset N_k$.  
    If $w \in L_0$, then since $L_0 \subset L_k$, any  $L_k$-empty circumsphere of 
    $\{x, w\}$ is also $L_0$-empty, giving $w \in \NN(x, \DelT(L_0))$.   Thus, by the way $Q$ is constructed, we have $\NN(x, \DelT(L_0)) \cap R = \varnothing$,
    so $w \in X_{ \gamma(x)} \subset N_k$.  Thus, $\NN(x, \DelT(L_k)) \subset N_k$, as desired.

 We now use~\cref{eq:star-invariant} to establish the result.  We first consider conflict pairs of the form $(\sigma,z_k)$, where $z_k\in R$ and $x\in \sigma$.   By~\Cref{eq:star-invariant},
 $\sigma$ is a simplex in $\DelT(N_k)$ if and only if it is a simplex in $\DelT(L_k)$.  Thus, the conflict pair $(\sigma,z_k)$ is recorded by one algorithm if and only if it is recorded by the other.

It remains to consider conflict pairs of the form $(\sigma,x)$.  Let 
\[N_{-1}=X_{\gamma(x)}{\subdown}\quad \textup{ and } \quad L_{-1}=X_{\gamma(x)}{\subdown}\cup (R\setminus Q).\]
A conflict pair of the form $(\sigma,x)$ arises in the non-local algorithm exactly when $\sigma\in \DelT(N_{-1}) \setminus \DelT(N_0)$, and in the local algorithm exactly when 
 $\sigma \in \DelT(L_{-1}) \setminus \DelT(L_0)$.  
        It therefore suffices to show that
        \[\DelT(L_{-1}) \setminus \DelT(L_0) = \DelT(N_{-1}) \setminus \DelT(N_0).\] 
        Note that, by the Bowyer-Watson algorithm, $\DelT(N_{-1}) \setminus \DelT(N_0)$ and $\St(x, \DelT(N_0))$
        triangulate the same star-convex region, as do $\DelT(L_{-1}) \setminus \DelT(L_0)$
        and $\St(x, \DelT(L_0))$. 
        By~\Cref{eq:star-invariant}, we have
        \[\St(x, \DelT(N_0)) = \St(x, \DelT(L_0)),\]
        so $\DelT(N_{-1}) \setminus \DelT(N_0)$ and $\DelT(L_{-1}) \setminus \DelT(L_0)$
        triangulate the same star-convex region on the same vertex set.
        Moreover, since $N_{-1} \subset L_{-1}$, every $d$-simplex $\sigma \in \DelT(L_{-1}) \setminus \DelT(L_0)$
        lies in $\DelT(N_{-1}) \setminus \DelT(N_0)$, as its circumsphere is $(L_{-1})$-empty
        and hence $(N_{-1})$-empty. 
        This gives
        \[\DelT(L_{-1}) \setminus \DelT(L_0) \subset \DelT(N_{-1}) \setminus \DelT(N_0).\]
        Since both $\DelT(L_{-1}) \setminus \DelT(L_0)$ and $\DelT(N_{-1}) \setminus \DelT(N_0)$
        triangulate the same region, we have:
        \[\DelT(L_{-1}) \setminus \DelT(L_0) = \DelT(N_{-1}) \setminus \DelT(N_0).\]        
    This concludes the proof.
\end{proof}

\section{Topological equivalence of trifiltrations}\label{sec:top_equiv}
Our main theoretical justification for computing Delaunay-(\v Cech) trifiltrations is the following:  

\begin{theorem}\label{thm:main} The trifiltrations $\Del(\gamma)$, $\DelCech(\gamma)$, and $\Offset(\gamma)$  are all weakly equivalent.
\end{theorem}

This section is devoted to the proof of \cref{thm:main}.  We begin by recalling the definition of weak equivalence; in homotopy theory, this is a standard notion of ``homotopical equivalence'' of diagrams of topological spaces.  
For a poset $\mathcal{P}$, a \deff{$\mathcal{P}$-indexed filtration} is a collection of topological spaces $F=(F_p)_{p \in \mathcal{P}}$ 
such that $F_p \subset F_q$, whenever $p \leq q$.  
Equivalently, $F$ is a functor from the poset category of $\mathcal{P}$ to the category of topological spaces $\Top$.   Note that each of $\Del(\gamma)$, $\DelCech(\gamma)$, and $\Offset(\gamma)$ is an $\R^2\times [0,\infty)$-indexed filtration, where $\R^2\times [0,\infty)$ is given the product partial order.  

\begin{definition}\label{def:top_equiv}\mbox{}
\begin{itemize}
\item[(i)] Given  $\mathcal{P}$-indexed filtrations $F,F'$, a \deff{natural transformation} $\psi\colon F\to F'$ is a collection 
 of continuous maps $(\psi_p:F_p \to F'_p)_{p \in \mathcal{P}}$ such that the following diagram commutes for all $p\leq q$:
 \[
\begin{tikzcd}
    F_p \arrow[r, hook] \arrow[d,"\psi_p",swap] & F_q \arrow[d,"\psi_q"] \\
    F'_p \arrow[r, hook] & F'_q
\end{tikzcd} \]
\item[(ii)] We call $\psi$ a \deff{pointwise homotopy equivalence} if each $\psi_p$ is a homotopy equivalence.  
\item[(iii)] We say that $F$ and $F'$ are \deff{weakly equivalent} and write $F\simeq F'$ if there exists a zigzag of pointwise homotopy equivalences connecting $F$ and $F'$:
\[
\begin{tikzcd}[row sep=small, column sep=small]
 & G^1 \arrow[dr] \arrow[dl] & & G^3  \arrow[dr] \arrow[dl] & & G^k  \arrow[dr] \arrow[dl] & \\
F  & & G^2 & & \cdots & & F'.
\end{tikzcd}
\]
\end{itemize}
\end{definition}

To show that  $\Del(\gamma)\simeq \Offset(\gamma)$, we characterize $\Del(\gamma)$ as the nerve of a cover of a certain trifiltration $\TelOffset$ in $\R^{d}\times\R^{2}$, which we call the \emph{telescopic offset trifiltration}.   A functorial version of the nerve theorem \cite{bauerUnifiedViewFunctorial2023} then implies that $\Del(\gamma)\simeq \TelOffset$.  In addition, a simple deformation retraction argument gives that $\TelOffset\simeq \Offset(\gamma)$.  Hence, $\Del(\gamma)\simeq \Offset(\gamma)$.  Similar constructions have previously been used to establish weak equivalence of 1-parameter and 2-parameter filtrations in~\cite{cavanna2015geometric, corbet2023computing}.  Our proof that $\DelCech(\gamma)\simeq \Offset(\gamma)$ is similar, and in particular, uses a description of $\DelCech(\gamma)$ as a nerve; this extends a description of 1-parameter Delaunay-\v Cech filtrations as nerve due to Blaser and Brun \cite{blaser2022relative}, which we review in \cref{Sec:WE_DelCech_1_param}.

\subparagraph{Functorial nerve theorem for closed, semi-algebraic sets.}
A \deff{cover}  of a set  $X\subset \R^d$ is a set $C = \{C(i)\}_{i \in I} $ of subsets of $X$ whose union is $X$.
The \deff{nerve} of $C$ is the simplicial complex \[ \Nrv(C) = \left\{ \sigma \subset I \textup{ finite and non-empty} :  \bigcap\limits_{i\in \sigma}{ C(i) } \neq \emptyset \right\}.\]
We say the cover $C$ is \deff{good} if every intersection of cover elements is either empty or contractible. 

These definitions extend naturally to $\mathcal{P}$-indexed filtrations, as follows: For $F$ a $\mathcal{P}$-indexed filtration, a \deff{cover} of $F$ is a collection of $\mathcal{P}$-indexed filtrations $\mathcal C = \{\mathcal C{(i)}\}_{i \in I} $, such that for each $p\in \mathcal P$, $\{\mathcal C(i)_p\}_{i\in I}$ is a cover of $F_p$.  We say that $\mathcal C $ is \deff{good} if for each $p$, $\{\mathcal C(i)_p\}_{i\in I}$ is a good cover of $F_p$. We define $\Nrv(\mathcal C)$, the \emph{nerve of $\mathcal C$}, to be the $\mathcal{P}$-indexed filtration given by 
$\Nrv(\mathcal C)_p=\Nrv(\{\mathcal C(i)_p\}_{i\in I})$, with structure maps the inclusions.

We will use the following functorial version of the nerve theorem for semi-algebraic sets; see also  \cite{bauerUnifiedViewFunctorial2023} for a thorough treatment of other variants of the nerve theorem.  

\begin{theorem}\label{Thm:Nerve_Semi_algebraic}
Let $\mathcal C=\{\mathcal C(i)\}_{i\in I}$ be a finite, good cover of a $\mathcal{P}$-indexed filtration $F$ such that 
for each $p\in \mathcal {P}$ and $i\in I$, $\mathcal{C}(i)_p$ is closed and semi-algebraic.  Then $F\simeq \Nrv(\mathcal C)$.  
\end{theorem}

\begin{proof}
This follows immediately from \cite[Theorem 5.9]{bauerUnifiedViewFunctorial2023}, using the fact that an inclusion of closed semi-algebraic sets satisfies the homotopy extension property \cite[Theorem 4]{delfs1984separation}.
\end{proof}

\begin{remark} \cref{Thm:Nerve_Semi_algebraic} is proven in \cite[Theorem 4.10]{bauerUnifiedViewFunctorial2023} under the additional assumption that each $\mathcal{C}(i)_p$ is compact, via a somewhat different argument.
\end{remark}

\subsection{Weak equivalence of Delaunay and offset trifiltrations}

\subparagraph{Telescopic offset trifiltration.} For $(p,r)\in\R^{2}\times \R^+$, the \deff{telescopic offset} $\TelOffset_{p,r} \subset \R^{d}\times \R^2$ is given by
\begin{equation}
  \TelOffset_{p,r} \coloneqq \bigcup_{q\leq p} \Offset(\gamma)_{q,r} \times \Set{q}.
\end{equation}
Equivalently, \[\TelOffset_{p,r}= \bigcup_{x\in X} \Ball_r(x)\times [\gamma(x),p],\] where $ [\gamma(x),p]=\{q\in \R^2\mid \gamma(x)\leq q\leq p\}$.  Varying $p$ and $r$, we obtain the \deff{telescopic offset trifiltration}
$\TelOffset\colon \R^{2}\times\R^{+} \to \Top$.  
Let $\psi\colon \TelOffset\to \Offset(\gamma)$ be the natural transformation induced by the projection of $\R^d\times \R^2$ onto the first $d$ coordinates. 
\begin{lemma}\label{prop:teloffset}
The natural transformation $\psi$ is a pointwise homotopy equivalence.
\end{lemma}
\begin{proof}
For $(p,r)\in \R^2\times \R^+$, $\psi_{p,r}$ is the map given by $(w,q)\mapsto w$.  This map is a homotopy equivalence with homotopy inverse $f$ given by $w\mapsto (w,p)$; indeed, $f \circ \psi_{p,r}$ is homotopic to the identity on $\TelOffset_{p,r}$ via the straight-line homotopy $H((w,q),t)=(w,tq+(1-t)p)$.
\end{proof}

\subparagraph{A cover by telescopic Voronoi balls.}
We next define our cover of $\TelOffset$.  
For $x\in X$, we define the \deff{telescopic Voronoi cell} 
\[
  \Tel(x) \coloneqq \operatorname{cl}(\bigcup\limits_{ \gamma(x)\leq q} {\Vor(x, X_{q}) \times \{q\}}),
\]
where $\operatorname{cl}$ denotes the closure.  
For $(p,r)\in \R^2\times \R^+$, we define the \deff{telescopic Voronoi ball} 
 \begin{align*}
  \TelVor(x)_{p,r} &\coloneqq \operatorname{cl}(\bigcup\limits_{ \gamma(x)\leq q\leq p} {\V(x, X_{q})_r \times \{q\}})\\
 &=  \Tel(x) \cap (\Ball_r(x)\times \{s\in \R^2\mid s\leq p\}).
 \end{align*}
See \Cref{fig:telvorcells} for an illustration in the case $d=1$.  
Allowing $p$ and $r$ to vary,  the spaces $\TelVor(x)_{p,r}$ assemble into a trifiltration $\TelVor(x)\colon \R^2\times \R^+\to \Top$.  We let $\TelVor=\{\TelVor(x)\}_{x\in X}$.

    \begin{figure}[h!]
    \centering
    \vspace{1em}
    \includegraphics[width=0.65\textwidth, page=1, trim=20pt 30pt 20pt 20pt]{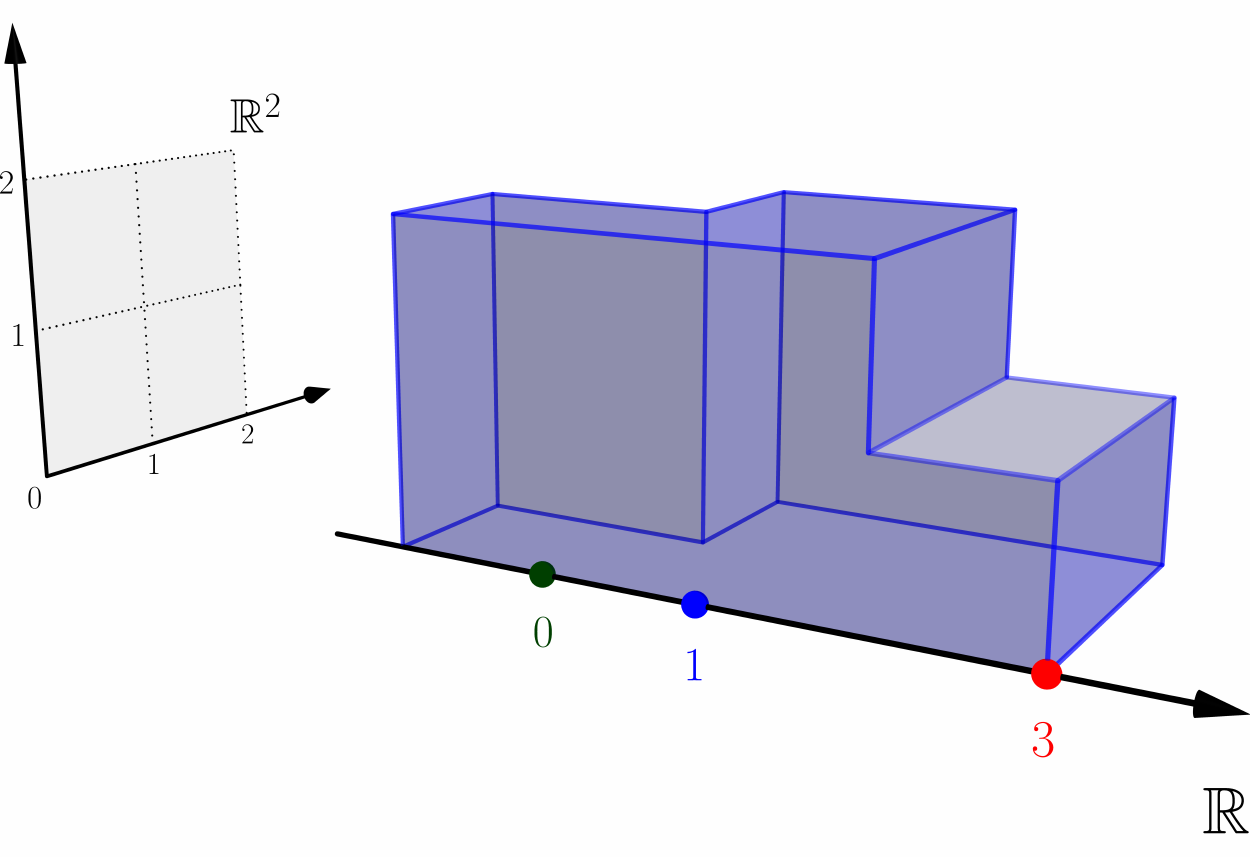} \\
    \vspace{1em}
%
%
%
%
%
    \caption{For $X=\{0,1,3\}$ with $\gamma(0) = (1,0)$, $\gamma(1) = (0,0)$, and $\gamma(3) = (0,1)$, the telescopic Voronoi ball $\TelVor(1)_{(2,2),2}$ is shown as the blue shaded polyhedron.}
        \label{fig:telvorcells}
\end{figure}

\begin{lemma} \label{lem:nrvd}
         The set $\TelVor$ is a good cover of $\TelOffset$. 
\end{lemma}
\begin{proof}
    For each $q \leq p$, $\{\V(x, X_{q})_{r}\}_{x\in X}$ forms a (good) cover of $\Offset(\gamma)_{q,r}$, so $\{\TelVor(x)_{p,r} \}_{x\in X}$ is a cover of $\TelOffset_{p,r}$.  Thus,  $\TelVor$ is a cover of $\TelOffset$. 
    
    It remains to show that this cover is good.  We need to show that if $\bigcap_{i=1}^k \TelVor_{p,r}(x_i)$ is non-empty, then it is contractible.  Let $\sigma=\{x_1,\ldots,x_k\}$ and for $q\in \R^2$, let $C_q= \bigcap_{i=1}^k \V(x_i,X_q)_r$.  Note that \[\gamma(\sigma)=\min \{q\in \R^2\mid C_q\ne \emptyset\}.\]  Moreover, if $\gamma(\sigma)\leq q\leq q'\leq p$, then $X_{q}\subset X_{q'}$ and therefore $\V(x_i,X_{q'})_r\subset \V(x_i,X_{q})_r$ for each $i$.  Hence, $C_{q'}\subset C_{q}$.  It follows that if 
    \[(w,q')\in \bigcap_{i=1}^k \TelVor(x_i)_{p,r},\] then 
    \[(w,q)\in \bigcap_{i=1}^k \TelVor(x_i)_{p,r}.\]  
    Therefore, we have a deformation retraction $H$ of  $\bigcap_{i=1}^k \TelVor(x_i)_{p,r}$ onto $C_{\gamma(\sigma)}\times \{\gamma(\sigma)\}$ given by $ H((w,q),t)=(w,(1-t)q+t\gamma(\sigma))$.   Note that $C_{\gamma(\sigma)}$ is contractible, since it is a non-empty intersection of Voronoi balls, hence convex.  Thus, $\bigcap_{i=1}^k \TelVor(x_i)_{p,r}$ is contractible.
\end{proof}

In what follows, we identify simplices of $\Nrv(\TelVor)$ with their  corresponding subsets of $X$. 
\begin{lemma}\label{thm:deleq}
The trifiltrations $\Nrv(\TelVor)$ and  $\Del(\gamma)$ are equal.  
\end{lemma}
\begin{proof}
For $\sigma\subset X$, let $x=\max_1(\sigma)$,  $y=\max_{2}(\sigma)$, and $\tau=\sigma\setminus \{x,y\}$.  
By definition, for $(p,r)\in \R^2\times \R^+$, we have $\sigma \in \Del(\gamma)_{p,r}$ if and only if  $\gamma(\sigma)\leq p$ and there exists a $(X_{\gamma(\sigma)}\setminus \{x,y\})$-empty circumsphere $S$ of $\tau$ with radius at most $r$ such that $x$ and $y$ lie inside or on $S$.  

We first consider the case that $x=y$.  Note that if $\sigma=\{x\}$, then $\sigma\in \Del(\gamma)_{p,r}$ exactly when $\gamma(\sigma)\leq p$, and also $\sigma\in \Nrv(\TelVor)_{p,r}$ exactly when $\gamma(\sigma)\leq p$.  Therefore, $\sigma\in  \Del(\gamma)_{p,r}$ if and only if  $\sigma\in \Nrv(\TelVor)_{p,r}$.  Now assume that  $\dim(\sigma)>0$.  Then $\sigma \in \Del(\gamma)_{p,r}$ is witnessed by a sphere $S$ with center $c$ and radius at most $r$ if and only if 
\[c\in \left(\bigcap_{z\in \tau} \V(z,X_{q})_r\right)\cap \V(x,X_{\gamma(x)})_r\]
for all $\gamma(\tau) \leq q< \gamma(x)$.  Since the elements of $\TelVor$ are closed sets, this holds if and only if 
\begin{equation}\label{eq:Tel_Vor}
(c,\gamma(x))\in \bigcap_{z\in \sigma} \TelVor(z)_{p,r}.
\end{equation}
 Moreover, $\sigma\in \Nrv(\TelVor_{p,r})$ if and only if there exists some $c$ satisfying \eqref{eq:Tel_Vor}.  Thus $\sigma\in  \Del(\gamma)_{p,r}$ if and only if $\sigma\in \Nrv(\TelVor_{p,r})$.

We next consider the case that $x\ne y$.  Note that if $\sigma=\{x,y\}$, then \cref{Prop:Empty_Tau_Delaunay} implies that $\sigma\in \Del(\gamma)_{p,r}$ exactly when $\sigma\in \Del(X_{p})_r$.  In addition, $\sigma\in \Nrv(\TelVor_{p,r})$ exactly when $\sigma\in \Del(X_{p})_r$.   Therefore, $\sigma\in  \Del(\gamma)_{p,r}$ if and only if  $\sigma\in \Nrv(\TelVor)_{p,r}$.  Now assume that  $\dim(\sigma)>1$.  Then $\sigma \in \Del(\gamma)_{p,r}$ is witnessed by a sphere $S$ with center $c$ and radius at most $r$ if and only if 
\[c\in \left(\bigcap_{z\in \tau} \V(z,X_{q})_r\right)\cap \V(x,X_{q'})_r \cap  \V(y,X_{q''})_r\]
for all $q$, $q'$, and $q''$ with 
\begin{align*}
\gamma(\tau) &\leq q< \gamma(\sigma),\\
(\gamma_1(x),\gamma_2(\tau))&\leq q'< \gamma(\sigma),\\
(\gamma_1(\tau),\gamma_2(y))&\leq q''< \gamma(\sigma).
\end{align*}
As above, since the elements of $\TelVor$ are closed sets, this holds if and only if $(c,\gamma(\sigma))\in \bigcap_{z\in \sigma} \TelVor(z)_{p,r}$, which in turn holds if and only if $\sigma \in  \Nrv(\TelVor_{p,r})$.  Thus, $\sigma\in  \Del(\gamma)_{p,r}$ if and only if $\sigma\in \Nrv(\TelVor_{p,r})$.
\end{proof}
\begin{proposition}\label{Prop:Del_Weak_Equiv} We have $\Del(\gamma)\simeq \Offset(\gamma)$.
\end{proposition}
\begin{proof}
For each $x\in X$, $\Tel(x)$ is a finite union of convex polyhedra, hence is semi-algebraic.   Thus, for each $(p,r)\in \R^2\times \R^+$, the space  $\TelVor(x)_{p,r}$ is an intersection of semi-algebraic sets, hence itself semi-algebraic.  Moreover, $\TelVor(x)_{p,r}$ is closed by construction.  Since $\TelVor$ is a good cover of $\TelOffset$  by \cref{lem:nrvd}, \cref{Thm:Nerve_Semi_algebraic} implies that $\TelOffset\simeq  \Nrv(\TelVor)$.  We also have $\mathcal O(\gamma) \simeq \TelOffset$ by \cref{prop:teloffset} and $\Nrv(\TelVor)= \Del(\gamma)$ by \cref{thm:deleq}, which yields the result.  
\end{proof}

\subsection{Weak equivalence of Delaunay-\v Cech and offset filtrations over $\R$}\label{Sec:WE_DelCech_1_param}
The weak equivalence of the 1-parameter Delaunay-\v Cech and offset filtrations $\DelCech(X)$ and $\Offset(X)$ was first proven by Bauer and Edelsbrunner \cite{bauer2017morse}, using a discrete Morse theory argument.  Subsequently, Blaser and Brun \cite{blaser2022relative} gave a different proof of this result via a clever nerve construction, and also extended this to a relative version.  We next present a streamlined variant of this proof (in the absolute case only), as our proof that $\DelCech(\gamma)\simeq \Offset(\gamma)$ draws heavily on this.

For $x\in X$ and $r\geq 0$, let $\TelVO(x)_r= \Ball_r(x) \times \Vor(x,X)\subset \R^d\times \R^d$, and let 
\begin{align*}
\mathcal{OU}_r &= \bigcup_{x\in X} \mathcal{U}(x)_r \\
               &= \{(w,v)\in \R^d\times \R^d \mid w\in \Offset(X)_r,\, v\in \Vor(x,X) \\
               &\phantom{{}= \qquad \qquad \qquad} \textup{for some }x\in X\textup{ with }\|x-w\|\leq r\}.
\end{align*}
Varying $r$, we obtain filtrations $\mathcal{OU}=(\mathcal{OU}_r)_{r\in \R^+}$ and $\TelVO(x)=(\TelVO(x)_r)_{r\in \R^+}$ for all  $x\in X$.  
The projection of $\R^d\times \R^d$ onto the first factor induces a natural transformation $\psi\colon \mathcal{OU}\to \Offset(X)$.

\begin{lemma}\label{Lem:Star_Convex}
For each $w\in \Offset(X)_r$, the level set $\psi^{-1}_r(w)$ is star-convex with center $(w,w)$.  
\end{lemma}

\begin{proof}
Let $A=\{x\in X\mid \|w-x\|\leq r\}$, let $B=X\setminus A$, and let \[V=\{v\in \R^d \mid v\in \Vor(x,X)\textup{ for some }x\in A\}.\]
  Then $\psi^{-1}_r(w)=\{w\}\times V$, so it suffices to show that $V$ is star-convex with center $w$.  Note that $w\in V$.  Suppose $w\ne v\in V$.   By definition, $v\in \Vor(x,X)$ for some $x\in X$ with $\|w-x\|\leq r$. Then for all $b\in B$, we have $\|w-x\|<\|w-b\|$ and $\|v-x\|\leq \|v-b\|$.  The space 
 \[H=\{h\in \R^d\mid \|h-x\|\leq \|h-b\|\}\] is a half-space, hence convex, so the entire line segment from $w$ to $v$ lies in $H$.  Thus, every point $v'$ on this line segment satisfies $\|v'-x\|\leq \|v'-b\|$ for all $b\in B$.  Since $x\in A$, it follows that $v'\in V$.  
\end{proof}

\begin{lemma}\label{Lem:PRojection_H_E_1_Parameter}
The natural transformation $\psi$ is a pointwise homotopy equivalence.
\end{lemma}

\begin{proof}
It suffices to check that for each $r\geq 0$, the map $f\colon  \Offset(X)_r\to \mathcal{OU}_r$ given by $f(w)=(w,w)$ is a homotopy inverse to $\psi_r$.   We use \cref{Lem:Star_Convex} to obtain a homotopy $H$ from $f\circ \psi_r$ to the identity on $\mathcal{OU}_r$, as follows: $H((w,z),t)=(w,tz+(1-t)w)$.  Since $\psi_r\circ f$ is the identity map on $\Offset(X)_r$, the result follows.
\end{proof}

\begin{theorem}[\cite{bauer2017morse,blaser2022relative}]\label{Thm:1_Param_Weak_Equiv}
We have $\DelCech(X)\simeq \Offset(X)$.
\end{theorem}
\begin{proof}
By construction,  $\{\TelVO(x)\}_{x\in X}$ is a cover of $\mathcal{OU}$.  Note that for all $r\in \R^{+}$, $\TelVO(x)_r$ is closed, convex, and semi-algebraic, because it is the product of two closed, convex, and semi-algebraic sets.  Since the cover is pointwise convex, it is good. \cref{Thm:Nerve_Semi_algebraic} therefore implies that $\mathcal{OU}$ is weakly equivalent to the nerve of this cover, which is easily checked to be equal to $\DelCech(X)$.  In addition, \cref{Lem:PRojection_H_E_1_Parameter} implies that $\mathcal{OU} \simeq \Offset(X)$.  Hence, $\DelCech(X)\simeq\Offset(X)$, as desired.
\end{proof}

\subsection{Weak equivalence of the Delaunay-\v Cech and offset trifiltrations}\label{def:teldelcech}

We prove $\DelCech(\gamma)\simeq \Offset(\gamma)$, essentially by combining the proofs of \cref{Prop:Del_Weak_Equiv,Thm:1_Param_Weak_Equiv}.   For $x\in X$ and $p\in \R^2$, let 
\begin{align*}
\TelVor(x)_{p,\infty}&=\colim \TelVor(x)_{p,-}=\Tel(x)\cap (\R^d\times \{s\in \R^2\mid s\leq p\}),\\
\TelDC(x)_{p,r}& =  (\Ball_r(x)\times \TelVor(x)_{p,\infty}).
\end{align*} 
 Note that $\TelDC(x)_{p,r}\subset (\R^d\times \R^d\times \R^2)$.  Varying $p$ and $r$, we obtain a trifiltration $\TelDC(x)$.  Let $\TelDC$ be the trifiltration given by 
\begin{align*}
\mathcal O\TelDC_{p,r}&=\bigcup_{x\in X} \TelDC(x)_{p,r}\\
                          &=\{(w,v,q)\in \R^d\times \R^d\times \R^2 \mid w\in \Offset(X)_r,\ q\leq p,\\
                          &\, \quad \qquad v\in \Vor(x,X_q)\textup{ for some }x\in X_q \textup{ with }\|x-w\|\leq r\}.
                      \end{align*}

Let $\psi\colon \mathcal O\TelDC(\gamma)\to \Offset(\gamma)$ be the natural transformation induced by the projection of $\R^d\times \R^d\times \R^2$ onto the first $d$ coordinates.
\begin{proposition}\label{prop:TelVOequiv} The natural transformation $\psi$ is a pointwise homotopy equivalence.  
\end{proposition}

\begin{proof}
For each $(p,r)\in \R^2\times \R^+$, define the map $f\colon \Offset(\gamma)_{p,r}\to \mathcal O\TelDC_{p,r}$ by $f(w)=(w,w,p)$.  We show that $f$ is a homotopy inverse of $\psi_{p,r}$ by identifying a homotopy from $f\circ \psi_{p,r}$ to the identity map on $\mathcal O\TelDC_{p,r}$.  First note that we have a straight line homotopy $H$ from $f\circ \psi_{p,r}$ to the map $g$ given by $g(w,v,q)=(w,w,q)$, because if $(w,w,q)\in \mathcal O\TelDC_{p,r}$, then $(w,w,q')\in \mathcal O\TelDC_{p,r}$ for all $q\leq q'\leq p$.  Further,  \cref{Lem:Star_Convex} furnishes a homotopy $H'$ from $g$ to the identity map on $\mathcal O\TelDC_{p,r}$, namely $H((w,v,q),t)=(w,(1-t)w+tv,q)$.   The concatenation of $H$ and $H'$ is a homotopy from $f\circ \psi_{p,r}$ to the identity on $\mathcal O\TelDC_{p,r}$.
\end{proof}

Letting $\TelDC=\{\TelDC(x)\}_{x\in X}$, we have by construction that $\TelDC$ is a cover of~$\mathcal O\TelDC$.

\begin{lemma}\label{lem:nrvdc}
The cover $\TelDC$ of $\mathcal O\TelDC$ is good.  
\end{lemma}

\begin{proof}
Projecting each element $\TelDC(x)$ of $\TelDC$ onto the first $d$ coordinates of $\R^d\times \R^d\times \R^2$ yields the good cover of $\mathcal O(\gamma)$ by closed balls centered at the points of $X$.  Similarly, projecting each element of $\TelDC$ onto the last $d+2$ coordinates of $\R^d\times \R^d\times \R^2$ yields a cover of $\bigcup_{x\in X} 
\TelVor(x)_{p,\infty}$ which is also good, by essentially the same argument that we used to prove \cref{lem:nrvd}.  By definition, each trifiltration $\TelDC(x)$ is the pointwise cartesian product of its images under these two projections.  Thus, since intersections and cartesian products commute, while products of contractible sets are contractible, it follows that $\TelDC$ is good.
\end{proof}

\begin{lemma}\label{Lem:Incr_Vs_Nerve_Over_P}
For all $p\in \R^2$, $\sigma\in \Nrv(\TelVor_{p,\infty})$ if and only if $\sigma\in \Incr$ and $\gamma(\sigma)\leq p$.
\end{lemma}

\begin{proof}
Note that $\sigma\in \Nrv(\TelVor_{p,\infty})$ if and only if $\sigma \in  \Nrv(\TelVor)_{p,r}$ for some $r$.  By \cref{thm:deleq}, $\sigma \in  \Nrv(\TelVor)_{p,r}$ if and only if $\sigma\in \Del(\gamma)_{p,r}$.  Moreover, $\sigma \in\Del(\gamma)_{p,r}$ for some $r$ if and only if $\sigma\in \Incr$ and $\gamma(\sigma)\leq p$.    This gives the result.  
\end{proof}

\begin{proposition} \label{thm:delcecheq}
     The trifiltrations $\Nrv(\TelDC)$ and  $\DelCech(\gamma)$ are equal.
 \end{proposition}

\begin{proof}
 Since cartesian products and intersections commute, $\sigma \in \Nrv(\TelDC)_{p,r}$ if and only if $\bigcap_{x\in \sigma}\Ball_r(x) \neq\emptyset$ and $\bigcap_{x\in \sigma}\TelVor(x)_{p,\infty} \neq\emptyset$.  The latter equation holds if and only if $\sigma \in \Nrv(\{\TelVor(x)_{p,\infty}\}_{x\in X})$, which  in turn holds exactly when  $\sigma\in \Incr$ and $\gamma(\sigma)\leq p$, by \cref{Lem:Incr_Vs_Nerve_Over_P}.  Given this, the condition $\bigcap_{x\in \sigma}\Ball_r(x) \neq\emptyset$ holds if and only if $\sigma \in \DelCech(\gamma)_{p,r}$.
\end{proof}

\begin{proposition}\label{Prop:Del_Cech_Equivalence}
We have $\DelCech(\gamma)\simeq \Offset(\gamma)$.
\end{proposition}

\begin{proof}
Note that for each $x\in X$, $\TelDC(x)$ is pointwise closed and semi-algebraic, since each space in $\TelDC(x)$ is a product of two closed, semi-algebraic sets. Since $\TelDC$ is a good cover of $\mathcal O\TelDC$  by \cref{lem:nrvdc}, \cref{Thm:Nerve_Semi_algebraic} implies that $\mathcal O\TelDC\simeq  \Nrv(\TelDC)$.  We also have $\mathcal O(\gamma) \simeq \mathcal O\TelDC$ by \cref{prop:TelVOequiv} and $\Nrv(\TelDC)= \DelCech(\gamma)$ by \cref{thm:delcecheq}, which yields the result.  
\end{proof}

\Cref{thm:main} follows immediately from \Cref{Prop:Del_Cech_Equivalence} and \cref{Prop:Del_Weak_Equiv}.

\section{Implementation}\label{sec:implement}
We have written a C++ program to compute the Delaunay-\v{C}ech trifiltration $\DelCech(\gamma)$ for a function $\gamma: X \to \R^2$, where $X\subset \R^d$. 
The code is
\href{https://bitbucket.org/mkerber/function_delaunay}{available on Bitbucket}\footnote{
\url{https://bitbucket.org/mkerber/function_delaunay} (\href{https://archive.softwareheritage.org/swh:1:dir:95fb3611456975ff368b3cfb5604f4bd9011264e;origin=https://bitbucket.org/mkerber/function_delaunay;visit=swh:1:snp:8ce3f15de54bedbe074deeba25703bdb373e71e3;anchor=swh:1:rev:a481f8d90af7a97795c467d81fb3db2cf952a58e}{archived here})}.  
The code builds upon a prior implementation of the algorithm from \cite{alonsoDelaunayBifiltrationsFunctions2024} for computing the Delaunay-\v{C}ech filtration of an $\R$-valued function.  The program accepts a text file as input, where each line represents the coordinate of a
point $x \in \R^d$ along with its function value $\gamma(x)$, and outputs a chain complex representation of 
$\DelCech(\gamma)$ in the \texttt{scc2020} format specified in
\cite{scc2020}.    Both the local and non-local algorithms are implemented.  
Our implementation computes the incremental Delaunay complex essentially as described in~\cref{sec:computation}.  However, there are minor differences, which arise because some simplifications that have been incorporated into this paper have not yet been implemented in the code.  
As in~\cite{alonsoDelaunayBifiltrationsFunctions2024},  simplices are stored in a
simplex tree \cite{boissonnat2014simplextree, gudhi:FilteredComplexes}, 
and smallest enclosing balls are computed in CGAL~\cite{cgal:fghhs-bv-24a}.

\subparagraph{Experiments.} The test suite we use, obtained from~\cite{kerber_2025_rxedk-qyq77}, consists of point cloud samples of the circle $S^1$ and unit square $[0, 1]^2$ in $\R^2$, and of the 2-sphere $S^2$, torus $S^1 \times S^1$, and unit cube $[0, 1]^3$ in $\R^3$, with 5\% noise drawn from a uniform distribution on a box and perturbation to ensure general position.  We consider point clouds of 500, 1000, 2000, 4000, 8000, 16000 points.

For each point cloud $X$ we computed the interlevel Delaunay-\v{C}ech trifiltration of four different functions $\delta: X \to \R$, using the local algorithm.  
The four functions are:
\begin{itemize}
    \item \emph{codensity},
    $\delta(x) = -\sum\limits_{y\neq x}{\exp{\left(-\frac{\|x-y\|^2}{\sigma^2}\right)}}$,
    with $\sigma$ chosen as the $0.1^{\mathrm{th}}$ percentile of the non-zero distances between points in $X$, 
    \item \emph{$L_1$-coeccentricity},
    $\delta (x) = -\sum\limits_{y\in X}{\frac{\|x-y\|}{|X|}}$,
    \item \emph{height}, $\delta(x) = x_d$, where $x_d$ is the last coordinate of $x$,
    \item \emph{random}, where for each $x\in X$, $\delta(x)$ is chosen uniformly at random from $[0,1]$.
\end{itemize}
All experiments were performed on a computer with an Intel Core i7-5960X CPU
@3.00GHz and 64GB of memory, running Ubuntu 20.04.6 LTS. The code was
compiled with \texttt{g++ 9.4.0}. 

The full results of the experiments are \href{https://doi.org/10.5281/zenodo.19227801}{available on Zenodo}\footnote{\url{https://doi.org/10.5281/zenodo.19227801}}.  A representative subset of the results is given in~\Cref{table:experiments}; this data is plotted in~\Cref{fig:growth_plots}.  
We observe that across all examples, the size of the incremental Delaunay complex and the memory usage grow nearly linearly as a function of the input size. In contrast, the runtime grows nearly quadratically.  Thus, our experiments indicate that our approach is memory efficient, and therefore that substantially larger computations should be feasible.

\cref{table:reduced} compares the performance of the local and non-local algorithms on three types of examples, with $|X|$ up to 4000.  While the local algorithm offers no improvement over the non-local algorithm in the first example type (interlevel trifiltrations), for the two other example types, the local algorithm is always faster by a factor of between 3 and 14, and the factor increases with the size of the data set.

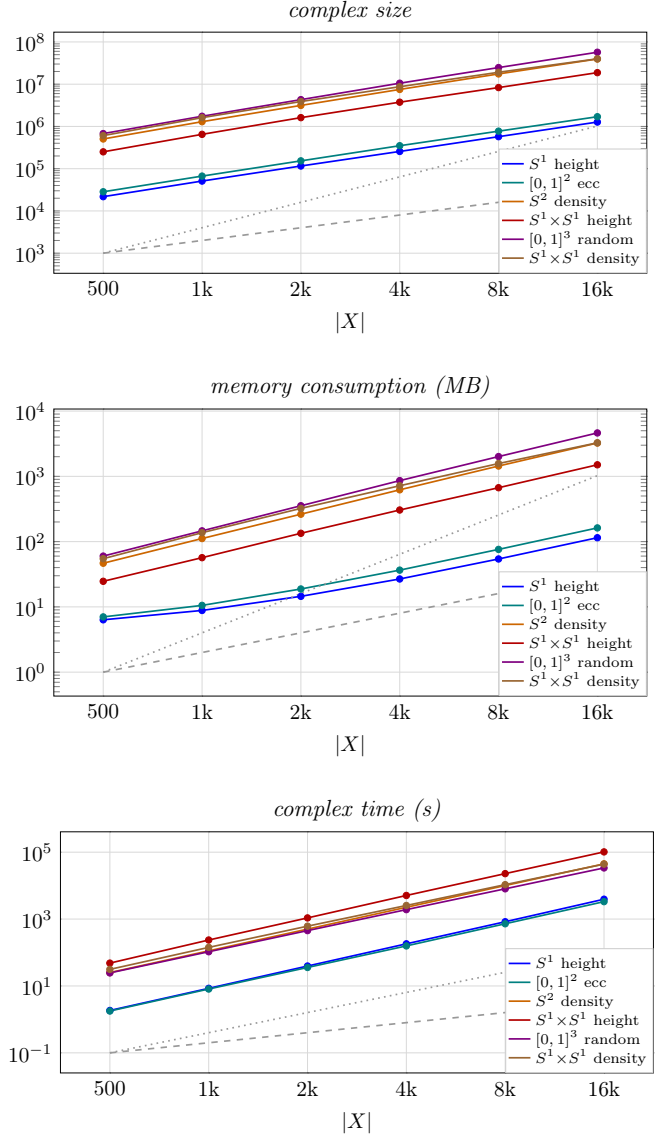
\begin{figure}[tp]
\centering
\begin{tikzpicture}[scale=0.8]
\begin{axis}[
  width=0.90\textwidth,
  height=0.44\textwidth, 
  xmode=log, ymode=log,
  title={\emph{complex size}},
  title style={yshift=-5pt},
  xlabel={$|X|$},
  xtick={500, 1000, 2000,4000,8000,16000},
  xticklabels={500, 1k, 2k,4k,8k,16k},
  scaled x ticks=false,
  x tick label style={font=\small},
  y tick label style={font=\small},
  xlabel style={font=\small},
  grid=major,
  major grid style={line width=0.3pt, draw=gray!30},
  legend style={font=\tiny, at={(1,0)}, anchor=south east,
                draw=gray!40, fill=white, row sep=-3pt, inner sep=1pt},
    legend image post style={scale=0.6},
  legend cell align=left,
  major tick length=0pt,
]
\addplot[thick, color=blue]
  coordinates {(500,21763)(1000,50967)(2000,115775)(4000,256043)(8000,572891)(16000,1270049)};
\addlegendentry{$S^1$ height}
\addplot[only marks, forget plot, mark=*, mark size=1.5pt, color=blue]
  coordinates {(500,21763)(1000,50967)(2000,115775)(4000,256043)(8000,572891)(16000,1270049)};

\addplot[thick, color=teal]
  coordinates {(500,28415)(1000,66775)(2000,153219)(4000,349689)(8000,773433)(16000,1698421)};
\addlegendentry{$[0,1]^2$ ecc}
\addplot[only marks, forget plot, mark=*, mark size=1.5pt, color=teal]
  coordinates {(500,28415)(1000,66775)(2000,153219)(4000,349689)(8000,773433)(16000,1698421)};

\addplot[thick, color=orange!80!black]
  coordinates {(500,505301)(1000,1295335)(2000,3145589)(4000,7514870)(8000,17501199)(16000,39523383)};
\addlegendentry{$S^2$ density}
\addplot[only marks, forget plot, mark=*, mark size=1.5pt, color=orange!80!black]
  coordinates {(500,505301)(1000,1295335)(2000,3145589)(4000,7514870)(8000,17501199)(16000,39523383)};

\addplot[thick, color=red!70!black]
  coordinates {(500,250593)(1000,650151)(2000,1615669)(4000,3751973)(8000,8316989)(16000,18764567)};
\addlegendentry{$S^1{\times}S^1$ height}
\addplot[only marks, forget plot, mark=*, mark size=1.5pt, color=red!70!black]
  coordinates {(500,250593)(1000,650151)(2000,1615669)(4000,3751973)(8000,8316989)(16000,18764567)};

\addplot[thick, color=violet]
  coordinates {(500,677493)(1000,1741733)(2000,4309321)(4000,10550531)(8000,24740775)(16000,57038595)};
\addlegendentry{$[0,1]^3$ random}
\addplot[only marks, forget plot, mark=*, mark size=1.5pt, color=violet]
  coordinates {(500,677493)(1000,1741733)(2000,4309321)(4000,10550531)(8000,24740775)(16000,57038595)};

\addplot[thick, color=brown!80!black]
  coordinates {(500,607060)(1000,1614337)(2000,3885310)(4000,8711827)(8000,19140397)(16000,39868943)};
\addlegendentry{$S^1{\times}S^1$ density}
\addplot[only marks, forget plot, mark=*, mark size=1.5pt, color=brown!80!black]
  coordinates {(500,607060)(1000,1614337)(2000,3885310)(4000,8711827)(8000,19140397)(16000,39868943)};

\addplot[black!40, no marks, thick, dotted, forget plot,
  domain=500:16000, samples=100] {0.004*x^2};
\addplot[black!40, no marks, thick, dashed, forget plot,
  domain=500:16000, samples=100] {2.0*x};
\addlegendentry{$O(n^2)$}
\end{axis}
\end{tikzpicture}
\par\bigskip 
\begin{tikzpicture}[scale=0.8]
\begin{axis}[
width=0.90\textwidth,
height=0.50\textwidth,
  xmode=log, ymode=log,
  title={\emph{memory consumption (MB)}},
  title style={yshift=-5pt},
  xlabel={$|X|$},
  xtick={500, 1000, 2000,4000,8000,16000},
  xticklabels={500, 1k, 2k,4k,8k,16k},
  scaled x ticks=false,
  x tick label style={font=\small},
  y tick label style={font=\small},
  xlabel style={font=\small},
  grid=major,
  major grid style={line width=0.3pt, draw=gray!30},
    legend style={font=\tiny, at={(1,0)}, anchor=south east,
                draw=gray!40, fill=white, row sep=-3pt, inner sep=1pt},
    legend image post style={scale=0.6},
  legend cell align=left,
  major tick length=0pt,
]
\addplot[thick, color=blue]
  coordinates {(500,6.34)(1000,8.82)(2000,14.51)(4000,26.76)(8000,54.20)(16000,114.66)};
\addlegendentry{$S^1$ height}
\addplot[only marks, forget plot, mark=*, mark size=1.5pt, color=blue]
  coordinates {(500,6.34)(1000,8.82)(2000,14.51)(4000,26.76)(8000,54.20)(16000,114.66)};

\addplot[thick, color=teal] 
  coordinates {(500,7.02)(1000,10.54)(2000,18.80)(4000,36.57)(8000,75.99)(16000,162.12)};
\addlegendentry{$[0,1]^2$ ecc}
\addplot[only marks, forget plot, mark=*, mark size=1.5pt, color=teal]
  coordinates {(500,7.02)(1000,10.54)(2000,18.80)(4000,36.57)(8000,75.99)(16000,162.12)};

\addplot[thick, color=orange!80!black]
  coordinates {(500,46.64)(1000,111.53)(2000,263.06)(4000,621.77)(8000,1439.11)(16000,3239.97)};
\addlegendentry{$S^2$ density}
\addplot[only marks, forget plot, mark=*, mark size=1.5pt, color=orange!80!black]
  coordinates {(500,46.64)(1000,111.53)(2000,263.06)(4000,621.77)(8000,1439.11)(16000,3239.97)};

\addplot[thick, color=red!70!black]
  coordinates {(500,24.66)(1000,56.75)(2000,134.14)(4000,304.95)(8000,669.48)(16000,1502.64)};
\addlegendentry{$S^1{\times}S^1$ height}
\addplot[only marks, forget plot, mark=*, mark size=1.5pt, color=red!70!black]
  coordinates {(500,24.66)(1000,56.75)(2000,134.14)(4000,304.95)(8000,669.48)(16000,1502.64)};

\addplot[thick, color=violet]
  coordinates {(500,60.00)(1000,146.14)(2000,354.44)(4000,859.78)(8000,2008.65)(16000,4623.33)};
\addlegendentry{$[0,1]^3$ random}
\addplot[only marks, forget plot, mark=*, mark size=1.5pt, color=violet]
  coordinates {(500,60.00)(1000,146.14)(2000,354.44)(4000,859.78)(8000,2008.65)(16000,4623.33)};

\addplot[thick, color=brown!80!black]
  coordinates {(500,54.73)(1000,136.96)(2000,322.87)(4000,717.71)(8000,1571.01)(16000,3265.80)};
\addlegendentry{$S^1{\times}S^1$ density}
\addplot[only marks, forget plot, mark=*, mark size=1.5pt, color=brown!80!black]
  coordinates {(500,54.73)(1000,136.96)(2000,322.87)(4000,717.71)(8000,1571.01)(16000,3265.80)};
\addplot[black!40, no marks, thick, dotted, forget plot,
  domain=500:16000, samples=100] {0.000004*x^2};
\addplot[black!40, no marks, thick, dashed, forget plot,
  domain=500:16000, samples=100] {0.002*x};

\addlegendentry{$O(n^2)$}
\end{axis}
\end{tikzpicture}
\par\bigskip 
\begin{tikzpicture}[scale=0.8]
\begin{axis}[
width=0.90\textwidth,
height=0.44\textwidth,
  xmode=log, ymode=log,
  title={\emph{complex time (s)}},
  title style={yshift=-5pt},
  xlabel={$|X|$},
  xtick={500, 1000, 2000,4000,8000,16000},
  xticklabels={500, 1k, 2k,4k,8k,16k},
  scaled x ticks=false,
  x tick label style={font=\small},
  y tick label style={font=\small},
  xlabel style={font=\small},
  grid=major,
  major grid style={line width=0.3pt, draw=gray!30},
  legend style={font=\tiny, at={(1,0)}, anchor=south east,
                draw=gray!40, fill=white, row sep=-3pt, inner sep=1pt},
  legend image post style={scale=0.6},
  legend cell align=left,
  major tick length=0pt,
]
\addplot[thick, color=blue]
  coordinates {(500,1.85)(1000,8.49)(2000,39.42)(4000,181.79)(8000,831.17)(16000,3924.57)};
\addlegendentry{$S^1$ height}
\addplot[only marks, forget plot, mark=*, mark size=1.5pt, color=blue]
  coordinates {(500,1.85)(1000,8.49)(2000,39.42)(4000,181.79)(8000,831.17)(16000,3924.57)};

\addplot[thick, color=teal]
  coordinates {(500,1.78)(1000,8.10)(2000,35.71)(4000,157.55)(8000,723.86)(16000,3319.30)};
\addlegendentry{$[0,1]^2$ ecc}
\addplot[only marks, forget plot, mark=*, mark size=1.5pt, color=teal]
  coordinates {(500,1.78)(1000,8.10)(2000,35.71)(4000,157.55)(8000,723.86)(16000,3319.30)};

\addplot[thick, color=orange!80!black]
  coordinates {(500,25.16)(1000,112.14)(2000,501.26)(4000,2249.95)(8000,10064.50)(16000,44737.40)};
\addlegendentry{$S^2$ density}
\addplot[only marks, forget plot, mark=*, mark size=1.5pt, color=orange!80!black]
  coordinates {(500,25.16)(1000,112.14)(2000,501.26)(4000,2249.95)(8000,10064.50)(16000,44737.40)};

\addplot[thick, color=red!70!black]
  coordinates {(500,48.22)(1000,236.34)(2000,1082.43)(4000,5075.66)(8000,22835.40)(16000,101974)};
\addlegendentry{$S^1{\times}S^1$ height}
\addplot[only marks, forget plot, mark=*, mark size=1.5pt, color=red!70!black]
  coordinates {(500,48.22)(1000,236.34)(2000,1082.43)(4000,5075.66)(8000,22835.40)(16000,101974)};

\addplot[thick, color=violet]
  coordinates {(500,24.56)(1000,105.52)(2000,451.97)(4000,1916.20)(8000,8030.43)(16000,33816.30)};
\addlegendentry{$[0,1]^3$ random}
\addplot[only marks, forget plot, mark=*, mark size=1.5pt, color=violet]
  coordinates {(500,24.56)(1000,105.52)(2000,451.97)(4000,1916.20)(8000,8030.43)(16000,33816.30)};

\addplot[thick, color=brown!80!black]
  coordinates {(500,31.32)(1000,142.33)(2000,613.11)(4000,2548.13)(8000,10684.30)(16000,43863.60)};
\addlegendentry{$S^1{\times}S^1$ density}
\addplot[only marks, forget plot, mark=*, mark size=1.5pt, color=brown!80!black]
  coordinates {(500,31.32)(1000,142.33)(2000,613.11)(4000,2548.13)(8000,10684.30)(16000,43863.60)};
\addplot[black!40, no marks, thick, dotted, forget plot,
  domain=500:16000, samples=100] {0.0000004*x^2};
\addplot[black!40, no marks, thick, dashed, forget plot,
  domain=500:16000, samples=100] {0.0002*x};

\addlegendentry{$O(n^2)$}
\end{axis}
\end{tikzpicture}
\caption{Log-log plots of complex size (top), memory consumption in MB (middle), and time in seconds (bottom) as a function of input size $|X|$, for the computations reported in \cref{table:experiments}.  For reference, in each plot, the graphs of a line $y=c_1x$ and quadratic $y=c_2x^2$ are shown as dashed and dotted lines, respectively, for some choices of $c_1$ and $c_2$ that vary among the three plots.  
In the top plot, for each type of example, the sizes ($y$-values) are well fit by a translate of the line representing $y=c_1x$, indicating that the sizes grow nearly linearly.  In the middle plot, the same is very nearly true, indicating that memory consumption also grows nearly  linearly.  
In the bottom plot, for each type of example, the runtimes ($y$-values) are well fit by a translate of the line representing $y=c_2x^2$, indicating that the runtimes grow nearly quadratically.}
\label{fig:growth_plots}
\end{figure}

\begin{table}[h!]
\includegraphics[width=\textwidth]{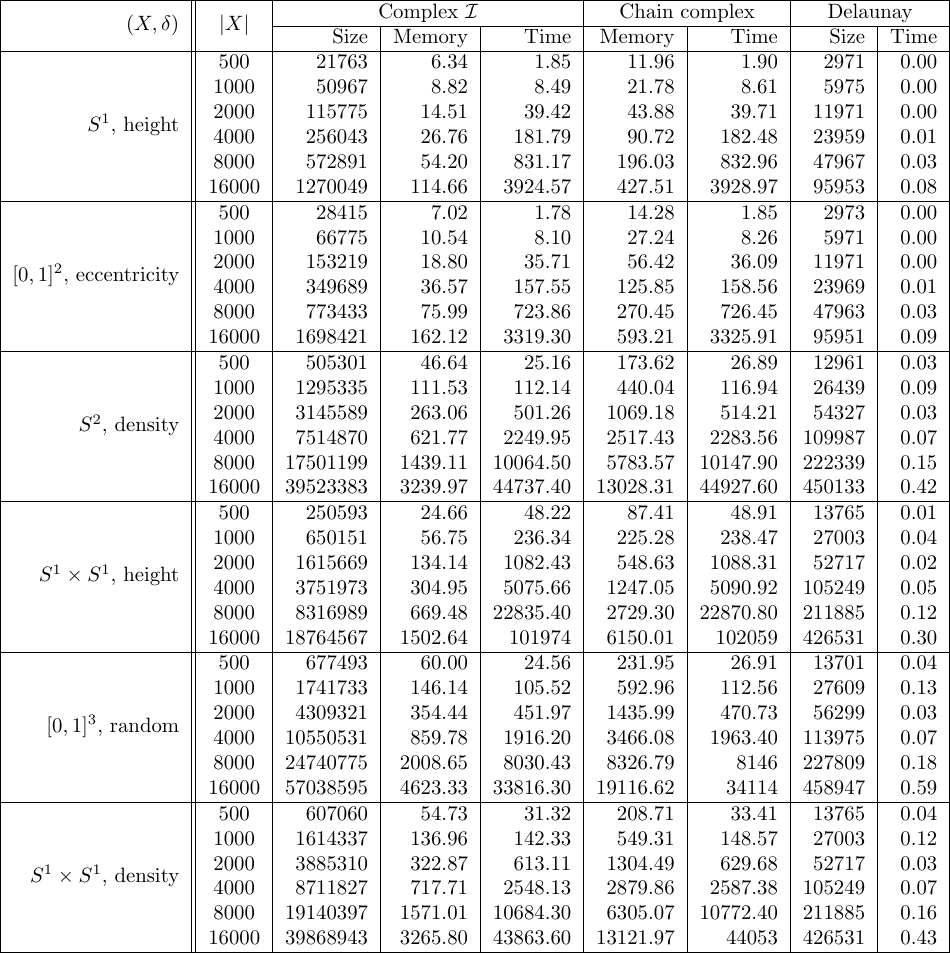}
\caption{Complex size, running time (in seconds), and the memory consumption (in
  MB) for computing the interlevel Delaunay-\v Cech trifiltration via the local algorithm. We also report the
  size of the Delaunay complex and the time required to compute it.
}
\label{table:experiments}
\end{table}

\begin{table}[h!]
\centering
\small
\includegraphics[width=0.7\textwidth]{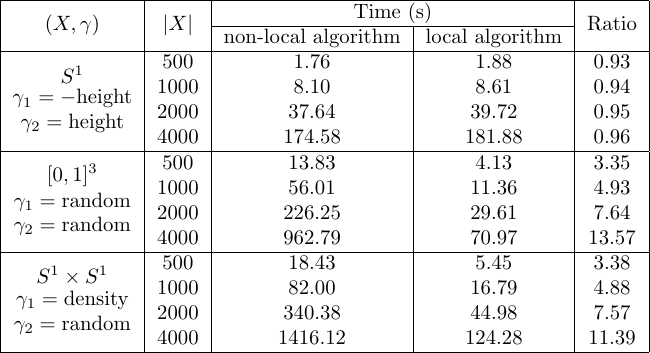}
\caption{Comparison of the non-local and local algorithms across several datasets.   In the interlevel case ($S^1$, $\gamma_1 = -\text{height}$, $\gamma_2 = \text{height}$), the local algorithm removes all points of $R$, but incurs a slight overhead from searching for vertices to remove, resulting in marginally worse running times than the non-local algorithm.  In the other cases, the local algorithm performs substantially better.}
\label{table:reduced}
\end{table}

\newpage
\appendix
\section{Computing incremental Delaunay radii}\label{sec:radiusDel}
For a point cloud $Z\subset \R^d$ in general position and $\sigma \in \DelT(Z)$, the \emph{Delaunay radius} of $\sigma$ is 
\begin{align*}
b_{\sigma}\coloneqq& \min\, \{r\in \R^+\mid \sigma\in \Del(Z)_r\}\\
               =&\min\, \{r\in \R^+\mid \exists \textup{ a $Z$-empty circumsphere of $\sigma$ with radius $r$}\}.
\end{align*}
Computing the Delaunay filtration $\Del(Z)$ amounts to computing both $\DelT(Z)$ and the Delaunay radii of all simplices in $\DelT(Z)$.  Assuming that $d$ is constant, a standard algorithm \cite{Edelsbrunner:alphashape,boissonnat2018geometric, gudhi:AlphaComplex} computes all Delaunay radii in time $O(|\DelT(Z)|)$ by iterating through the simplices in order of decreasing dimension.  In this section, we show that a slight variant of this algorithm computes the incremental Delaunay radii of all simplices of $\Incr$ in time $O(|\Incr|)$.

\subsection{Computation of ordinary Delaunay radii}\label{Sec:Ordinary_Del_Radii}
We begin by reviewing the standard algorithm to compute the Delaunay radii $(b_{\sigma})_{\sigma\in \DelT(Z)}$.  For $\sigma \in \DelT(Z)$, let $R_{\sigma}$ denote its smallest circumsphere.  It is a standard fact that $R_{\sigma}$ can be computed in constant time by solving a linear system.  We say that $\sigma$ is \deff{Gabriel}  if $R_{\sigma}$ is $Z$-empty. 
\begin{lemma}\label{Lem:Cofacet_Radii}
If $\sigma \in \DelT(Z)$ is not Gabriel, then $\sigma$ has a cofacet $\sigma \cup \{w\} \in \DelT(Z)$ such that $w$ lies inside $R_{\sigma}$.
\end{lemma}

\begin{proof}
Let $\hat c$ be the center of $R_{\sigma}$ and let $V=\bigcap_{z\in \sigma} \Vor(z,Z)$.  Since $\sigma$ is not Gabriel, $\hat c \notin V$. Let $v$ be the closest point of $V$ to $\hat c$. Since $\hat c \notin V$, the point $v$ lies on the boundary $\partial V$.  Thus, there exist points $w \in Z$ and $z\in \sigma$ such that $\|v-w\|=\|v-z\|$ and $\|b-w\|<\|b-z\|$ for some point $b$ lying on the open ray $L\coloneqq \{v+t(\hat c-v)\mid t>0\}$.  The set $A$ of points equidistant from all points of $\sigma$ is an affine subspace of $\R^d$, so since $v,\hat c\in A$, we also have $b\in A$.  Therefore, $w\not\in\sigma$.  Since $\|v-w\|=\|v-z\|$, we have $v \in V\cap \Vor(w,Z)$, so $\sigma \cup \{w\} \in \DelT(Z)$.  Since $\|b-w\|< \|b-z\|$, the line containing $L$ intersects the hyperplane $\{c\in \R^d \mid \|c-w\|=\|c-z\|\}$ only at $v$, so by the intermediate value theorem, we must have  $\|c-w\|<\|c-z\|$ for all $c\in L$.  In particular, we have  $\|\hat c-w\|<\|\hat c-z\|$, so $w$ lies inside of $R_{\sigma}$.
\end{proof}

By \cref{Lem:Cofacet_Radii}, we can check whether a simplex $\sigma\in \DelT(Z)$ is Gabriel in time proportional to the number of cofacets of $\sigma$.  Using this, we obtain the standard algorithm to compute the Delaunay radii $(b_{\sigma})_{\sigma\in \DelT(Z)}$, which we summarize in a single sentence: For each $\sigma \in \DelT(Z)$ in order of decreasing dimension, set 
\[
b_\sigma \leftarrow
\begin{cases}
\textup{radius of } R_{\sigma} & \text{if } \sigma \text{ is Gabriel,} \\
\min\, \{ b_\mu\mid \mu \text{ a cofacet of } \sigma\} & \text{otherwise.}
\end{cases}
\]
Since $d$ is assumed to be constant, the total number of facet-cofacet pairs of $\DelT(Z)$ is $O(|\DelT(Z)|)$.  Hence, this algorithm runs in time $O(|\DelT(Z)|)$.  The correctness of the algorithm follows from the following result.

\begin{proposition} \label{Prop:Del_Rad_Coface_Formula}
If $\sigma \in \DelT(Z)$ is not Gabriel, then \[b_\sigma = \min\, \{ b_\mu\mid \mu \textup{ a cofacet of } \sigma\}.\]
\end{proposition}

\begin{proof}
By definition, $b_{\sigma}$ is the radius of the smallest $Z$-empty circumsphere of $\sigma$.  Fixing $z\in \sigma$, the center $c^*$ of this sphere is the minimizer of the following optimization problem \[
\minimize_{c \in A} \|c - z\|^2 \quad \text{subject to} \quad \|c - w\| \geq \|c - z\| \quad \text{for all } w \in Z \setminus \sigma,
\] 
where $A$ is the set of points equidistant to all points of $\sigma$.   Note that the objective function of this problem is strictly convex and that $A$ is an affine subspace of $\R^d$, as for each $w\in \sigma$, the equality $\|c - w\| = \|c - z\|$ is equivalent to the equality $2\langle c, z- w \rangle+ \|w\|^2 - \|z\|^2  = 0$, which is linear in $c$.  Similarly, each inequality constraint of the optimization problem is equivalent to a linear one.  Thus, the feasible set $F$ of this optimization problem is a closed, convex polytope.    

We claim that $c^*\in \partial F$. To see this, let $\hat c$ denote the center of $R_{\sigma}$, i.e., $\hat c$ is the minimizer over $A$ of the corresponding unconstrained optimization problem.  Since $\sigma$ is not Gabriel by assumption, $\hat c\not \in F$.   If $c^*$ were an element of the interior of $F$, then one could move $c^*$ slightly toward $\hat c$ and remain in $F$ while strictly decreasing the objective, a contradiction.  Thus, we indeed have $c^*\in \partial F$, as claimed.  Therefore, at least one constraint is tight at $c^*$, i.e., $\|c^* - w\| = \|c^* - z\|$ for some $w \in Z \setminus \sigma$.  We then have $\nu\coloneqq \sigma\cup \{w\}\in \DelT(Z)$.

For any cofacet $\mu$ of $\sigma$, we have  $b_{\sigma}\leq b_{\mu}$.  In particular, $b_{\sigma}\leq b_{\nu}$.  The sphere with center $c^*$ and radius $b_{\sigma}$ is the smallest $Z$-empty circumsphere of $\sigma$ and also a circumsphere of $\nu$, so in fact $ b_{\sigma}=b_{\nu}$.  The result follows.
\end{proof}

\subsection{Computation of incremental Delaunay radii}
We next modify the above algorithm to compute the incremental Delaunay radii of
all simplices in $\Incr$.  For $\sigma \in \Incr$, let  $x=\max_1(\sigma)$,
$y=\max_2(\sigma)$, and $\tau=\sigma \setminus \{x, y\}$.  Let $S_{\sigma}$
denote the smallest circumsphere of $\tau$ such that $x$ and $y$ both lie on or
inside $S_{\sigma}$.  Since $\dim(\sigma)\leq d+2$, the sphere $S_{\sigma}$ can be computed using Welzl's
miniball algorithm~\cite{welzlSmallestEnclosingDisks1991} in worst-case constant time.  We say $\sigma$ is \deff{$\Incr$-Gabriel} if $S_{\sigma}$ is
$(X_{\gamma(\sigma)}\setminus \{x,y\})$-empty.  

The following is the incremental analogue of \cref{Lem:Cofacet_Radii}, and has a similar proof.

\begin{lemma}\label{Lem:Incremental_cofacet}
If $\sigma \in \Incr$ is not $\Incr$-Gabriel, then $\sigma$ has a cofacet $\sigma \cup \{w\} \in \Incr$ such that $w$ lies inside $S_{\sigma}$.  
\end{lemma}

\begin{proof}
We first assume that $\tau\ne \emptyset$.  Let $\hat c$ be the center of $S_{\sigma}$.  Let 
\begin{align*}
U&=\bigcap_{z\in \tau} \Vor(z,X_{\gamma(\sigma)}\setminus \{x,y\}),\\ 
U'&=\Vor(x,\tau\cup\{x\})\cap  \Vor(y,\tau\cup\{y\}),
\end{align*}
and $V=U\cap U'$.  Note that since $\sigma\in \Incr$, we have $V\ne \emptyset$.  Moreover, since $\sigma$ is not $\Incr$-Gabriel, $\hat c\notin U$, hence $\hat c\not\in V$.     Let $v$ be the closest point of $V$ to $\hat c$. Since $\hat c \notin V$, we have $v\in \partial V$.    But since $v,\hat c\in U'$ and $U'$ is convex, every point on the line segment from $v$ to $\hat c$ is in $U'$, so $v\in \partial U$.  Thus, there exist points $w \in X_{\gamma(\sigma)}\setminus\{x,y\}$ and $z\in \tau$ such that $\|v-w\|=\|v-z\|$ and $\|b-w\|<\|b-z\|$ for some point $b$ lying on the open ray $L\coloneqq \{v+t(\hat c-v)\mid t>0\}$.  The set $A$ of points equidistant from all points of $\tau$ is an affine subspace of $\R^d$, so since $v,\hat c\in A$, we also have $b\in A$.  Therefore, $w\not\in\tau$, hence $w\not\in \sigma$.  Since $\|v-w\|=\|v-z\|$, we have $v \in V\cap \Vor(w,X_{\gamma(\sigma)}\setminus \{x,y\})$, so $\sigma \cup \{w\} \in \Incr$.  Since $\|b-w\|<\|b-z\|$, the line containing $L$ intersects the hyperplane $\{c\in \R^d \mid \|c-w\|= \|c-z\|\}$ only at $v$, so by the intermediate value theorem, we must have  $\|c-w\|<\|c-z\|$ for all $c$ on $L$.  In particular, we have  $\|\hat c-w\|<\|\hat c-z\|$, so $w$ lies inside of $S_{\sigma}$.
 
Now assume that $\tau= \emptyset$.  We use the notation of \cref{Sec:Ordinary_Del_Radii}, taking the ambient point set $Z$ to be $X_{\gamma(\sigma)}$.  Note that $x\ne y$, since otherwise $\sigma$ would be $\Incr$-Gabriel.   \cref{Prop:Empty_Tau_Delaunay} implies that $x$ and $y$ both lie on $S_\sigma$, so $S_{\sigma}=R_{\sigma}$.   \cref{Prop:Empty_Tau_Delaunay} also tells us that $\sigma\in \Del(Z)$, so since $\sigma$ is not $\Incr$-Gabriel, it follows that $\sigma$ is not Gabriel.  By \cref{Lem:Cofacet_Radii}, $\sigma$ has a cofacet $\sigma \cup \{w\} \in \Del(Z)$ such that $w$ lies inside $R_{\sigma}=S_{\sigma}$.  Since $\Del(Z)\subset \Incr$, the result follows. 
\end{proof}

In analogy with the above, \cref{Lem:Incremental_cofacet} implies that we can check whether a simplex $\sigma\in \Incr$ is $\Incr$-Gabriel in time proportional to the number of cofacets of $\sigma$.   In particular, it tells us that every maximal simplex of $\Incr$ is $\Incr$-Gabriel.  This yields the following algorithm to compute the Delaunay radii of all simplices in $\Incr$: For each $\sigma \in \Incr$ in order of decreasing dimension, set 
\[
\omega_{\sigma} \leftarrow
\begin{cases}
\text{radius of } S_\sigma & \text{if } \sigma \text{ is $\Incr$-Gabriel,} \\
\min\, \{\omega_{\mu} \mid \mu\in \Incr \textup{ a cofacet of } \sigma\} & \text{otherwise.}
\end{cases}
\]
By the same argument as in the non-incremental case, this algorithm runs in time $O(|\Incr|)$.  In further analogy to the non-incremental case, the 
proof of correctness hinges on the following result.

\begin{proposition}\label{Prop:Incr_Cofact_Formula}
If $\sigma \in \Incr$ is not $\Incr$-Gabriel, then \[\omega_\sigma =\min\, \{\omega_{\mu} \mid \mu\in \Incr \textup{ a cofacet of } \sigma\} .\]
\end{proposition}

\begin{proof}
The proof is a variant of the proof of \cref{Prop:Del_Rad_Coface_Formula}.  Let $S$ denote the smallest witness of $\sigma$.  Note that we may have $S\ne S_{\sigma}$, since $S$ must be $X_{\gamma(\sigma)}\setminus \{x,y\}$-empty.  Recall that  $\omega_{\sigma}$ is the radius of $S$.  

We first assume that $\tau\ne \emptyset$.  Let $c^*$ and $\hat c$ denote the respective centers of $S$ and $S_{\sigma}$. Recall that $c^*$ is the minimizer of the strictly convex, linearly constrained optimization problem \eqref{eq:opt_Prob} from the proof of \cref{Prop:Min_Rad_Witness}.  
Let $F$ denote the feasible set of this optimization problem, and note that $F$ is a closed, convex (possibly unbounded) polytope.  Since $\sigma$ is not $\Incr$-Gabriel, we have $\hat c\not \in F$.  Note that $\hat c$ is the minimizer of the optimization problem obtained from \eqref{eq:opt_Prob} by dropping the constraints 
\[  \|c - w\| \geq \|c - z\|
            \text{ for all } w \in X_{\gamma(\sigma)} \setminus \sigma.\]
Thus, arguing as in the proof of \cref{Prop:Del_Rad_Coface_Formula}, we find  that $c^*\in \partial F$.  Therefore, at least one constraint of the optimization problem \eqref{eq:opt_Prob} is tight at $c^*$.  Since $\hat c$ satisfies the first two constraints, which are linear, it follows that we in fact have a tight constraint of the form $\|c^* - w\| = \|c^* - z\|$ for some $w \in X_{\gamma(\sigma)} \setminus \sigma$.  Thus, $w$ lies on $S$.  Letting $\nu\coloneqq \sigma\cup\{w\}$, we have $\nu\in \Incr$, since $S$ is a witness of $\nu$.  

For $\mu\in \Incr$ any cofacet of $\sigma$, we have  $\omega_{\sigma}\leq \omega_{\mu}$.  In particular, $\omega_{\sigma}\leq \omega_{\nu}$.  Since $S$ is the smallest witness of $\sigma$ and also a witness of $\nu$, in fact $ \omega_{\sigma}=\omega_{\nu}$.  Thus, $\omega_\sigma = {\min\, \{ \omega_\mu\mid \mu \textup{ a cofacet of } \sigma\}}.$

Now assume that $\tau=\emptyset$.  We again use the notation of \cref{Sec:Ordinary_Del_Radii}, taking $Z=X_{\gamma(\sigma)}$.  
As noted in the proof of \cref{Lem:Incremental_cofacet}, $\sigma\in \DelT(Z)$ and $\sigma$ is not Gabriel.
Moreover, by \cref{Prop:Empty_Tau_Delaunay}, $x$ and $y$ both lie on $S$, so $\omega_{\sigma}=b_{\sigma}$.   Therefore, \cref{Prop:Del_Rad_Coface_Formula} implies that $b_{\sigma}=b_{\nu}$ for $\nu\in \DelT(Z)$ some cofacet of $\sigma$.  Since  $\DelT(Z)\subset \Incr$, we have $\nu\in \Incr$.  We also have $b_\nu=\omega_\nu$, because otherwise $b_\sigma=\omega_\sigma \leq \omega_\nu < b_\nu$, a contradiction.  Thus, $\omega_\sigma=b_\sigma=b_\nu=\omega_\nu$.  Since $\omega_{\sigma}\leq \omega_\mu$ for  $\mu\in \Incr$ any cofacet of $\sigma$, this shows that $\omega_{\sigma}= \min\, \{\omega_{\mu} \mid \mu\in \Incr \textup{ a cofacet of } \sigma\}.$
\end{proof}

\begin{proposition} 
    The above algorithm correctly computes the incremental Delaunay radius $\omega_\sigma$ of each simplex $\sigma \in \Incr$.
\end{proposition}

\begin{proof}
We proceed by downward induction on the face poset of $\sigma$.  In the base case that $\dim(\sigma)$ is maximal, $\sigma$ is $\Incr$-Gabriel, as noted above, so  $\omega_{\sigma}$ is the radius of $ S_{\sigma}$.  
For the induction step, assume that $\sigma$ is not top-dimensional and that $\omega_{\mu}$ has been correctly computed for all cofacets $\mu$ of $\sigma$.  As above, if $\sigma$ is $\Incr$-Gabriel, then $\omega_{\sigma}$ is the radius of $S_\sigma$.  If $\sigma$ is not $\Incr$-Gabriel, then in view of the induction hypothesis, \cref{Prop:Incr_Cofact_Formula} implies that the algorithm correctly computes $\omega_{\sigma}$. 
\end{proof}

\newpage

\bibliographystyle{plainurl}
\bibliography{refs}
\end{document}